\documentclass[submission,copyright,creativecommons]{eptcs}
\providecommand{\event}{Journal of Logical and Algebraic Methods in Programming}

\usepackage{multibib}

\usepackage{amssymb}

\usepackage{url}
\usepackage{diagrams}

\usepackage{bussproofs} 

\usepackage[cmex10]{amsmath}
\usepackage[english]{babel}
\usepackage[latin1]{inputenc}
\usepackage{longtable}
\usepackage{color}
\usepackage{graphicx}
\usepackage{fontenc}
\usepackage{url}
\usepackage{psfrag}
\usepackage{listings}
\usepackage{bm}
\usepackage{hyperref}

\usepackage{latexsym,epsfig}

\usepackage{xifthen}
\usepackage{xspace}

 \usepackage[margin=20mm,top=20mm,bottom=30mm]{geometry} 

\newtheorem{theorem}{Theorem}[section]

\newtheorem{proposition}[theorem]{Proposition}
\newtheorem{corollary}[theorem]{Corollary}

\newtheorem{definition}[theorem]{Definition}

\newtheorem{notation}[theorem]{Notation}
\newtheorem{remark}[theorem]{Remark}
\newtheorem{myexample}[theorem]{Example}

\newenvironment{proof}[1][\!\!\,]{\vspace{1ex}\noindent\textbf{Proof #1: }}{\hfill$\Box$\vspace{2ex}}

\lstnewenvironment{Creol}{%
  \lstset{language=Java}
   \lstset{basicstyle=\ttfamily\small} 
   \lstset{keywordstyle=\bf\small}
   \lstset{commentstyle=\sffamily}
   \lstset{tabsize=2}
\lstset{mathescape=true}
   \lstset{columns=flexible}
   \lstset{morekeywords={interface,inv,then,fi,intseq,begin,end,with,inherits,in,out,op,await,var,
     in,all,upgrade}
     }
\csname lst@SetFirstLabel\endcsname}
{\csname lst@SaveFirstLabel\endcsname}

\newcommand{\kw}[1]{\textbf{#1}}

\newcounter{case}

\newcommand{\oomit}[1]{}

\newcommand{\bs}[1]{}

\newcommand{\intechrep}[1]{}

\newcommand\cat[1]{\ensuremath{\mathbb{#1}}}
\newcommand\barecat[1]{\ensuremath{\underline{\mathbb{#1}}}}
\newcommand\basic[1]{\ensuremath{#1^{b}}}
\newcommand\discrete[1]{\ensuremath{#1^{\!d}}}
\newcommand\objects[1]{\ensuremath{|#1|}}
\newcommand\morphisms[1]{\ensuremath{\mathit{Mor}(#1)}}
\newcommand\source[1]{\ensuremath{#1^{s}}}
\newcommand\target[1]{\ensuremath{#1^{t}}}
\newcommand\composition{\ensuremath{\circ}}
\newcommand\labtrans[2]{%
  \ifthenelse{\isempty{#1}}%
    {\ensuremath{\mathbf{LT}}}
    {\ensuremath{\mathbf{LT}(#1,#2)}}
}
\newcommand\ulabtrans[2]{%
  \ifthenelse{\isempty{#1}}%
    {\ensuremath{\mathbf{ULT}}}
    {\ensuremath{\mathbf{ULT}(#1,#2)}}
}

\newcommand\encapsulate[2]{\ensuremath{\mathbf{Enc}(#1,#2)}}
\newcommand\prodcat{\ensuremath{\times}}
\newcommand\getcat{\ensuremath{\mathit{get}}}
\newcommand\respectsDependencies{\ensuremath{\subseteq}}
\newcommand\trivcat{\ensuremath{\mathbf{1}}}
\newcommand\indexE{\ensuremath{\mathit{I}_{L}}}
\newcommand\indexU{\ensuremath{\mathit{I}_{U}}}
\newcommand\idFunc{\ensuremath{\mathit{IdFun}}}
\newcommand\idRecord{\ensuremath{\mathit{IdRec}}}
\newcommand\idRecordLab{\ensuremath{\mathit{IdRecLab}}}
\newcommand\idObj{\ensuremath{\mathit{IdObj}}}
\newcommand\idFutures{\ensuremath{\mathit{IdFut}}}
\newcommand\idClass{\ensuremath{\mathit{IdClass}}}
\newcommand\idMethods{\ensuremath{\mathit{IdMethods}}}

\newcommand\idType{\ensuremath{\mathit{IdType}}}
\newcommand\catendo[1]{\ensuremath{\cat{E}nd(#1)}}
\newcommand\naturalTransf{\ensuremath{\eta}}
\newcommand\AttributesSort{\ensuremath{\mathit{At}}}

\newcommand\SigmaExec{\ensuremath{\Sigma}}
\newcommand\SigmaUpdate{\ensuremath{\Sigma_{upd}}}
\newcommand\metaVariables{\ensuremath{\mathsf{Var}}}
\newcommand\arity[1]{\ensuremath{\mathit{ar}(#1)}}

\newcommand\sortedterms[2]{\ensuremath{\mathsf{Terms_{#1}}(#2)}}
\newcommand\skipStatement{\ensuremath{\mathbf{skip}}}
\newcommand\nilValue{\ensuremath{\mathsf{nil}}}
\newcommand\sequence{\ensuremath{\,;}}
\newcommand\letStatement[2]{\ensuremath{\,\mathbf{let}\,#1\,\mathbf{in}\,#2\,}}
\newcommand\ifthenelseStatement[3]{\ensuremath{\,\mathbf{if}\,#1\,\mathbf{then}\,#2\,\mathbf{else}\,#3\,}}
\newcommand\funcdecl[3]{\ensuremath{\,\mathbf{fun\,#1}(#2)\,\{#3\}\,}}
\newcommand\recorddecl[2]{\ensuremath{\,\mathbf{record\ #1}\,\{#2\}\,}}
\newcommand\typedecl[2]{\ensuremath{\,\mathbf{type\ #1} = #2\,}}
\newcommand{\objectTerm}[2]{\ensuremath{\,\langle\,\mathbf{#1}\,\mid\,#2\,\rangle\,}}
\newcommand{\restlabel}{\ensuremath{\,\dots\,}}
\newcommand{\proteus}{\textsc{Proteus}\xspace}
\newcommand{\proteusSrcLang}{\textsc{Proteus}\xspace}
\newcommand{\creol}{\textsc{Creol}\xspace}
\newcommand{\stump}{\textsc{Stump}}
\newcommand{\upgradej}{\textsc{UpgradeJ}}
\newcommand{\coqprover}{Coq}
\newcommand{\yieldThread}{\ensuremath{\mathbf{yield}}}
\newcommand{\methodend}[1]{\ ?(#1)}

\newcommand{\asyncThread}[1]{\ensuremath{\mathbf{async}\,(#1)}}

\newcommand{\evalContext}{\ensuremath{\mathit{Ev}}}
\newcommand{\holeContext}[1]{\ensuremath{[#1]}}
\newcommand{\evalContextFill}[1]{\ensuremath{\evalContext\holeContext{#1}}}
\newcommand{\insertThread}{\ensuremath{\oplus}}
\newcommand{\deleteThread}{\ensuremath{\ominus}}
\newcommand{\callMethodOfIn}[3]{\ensuremath{\,\mathbf{#3}!#2.#1\,}}
\newcommand{\readMethodReturnFromIn}[2]{\ensuremath{\,\mathbf{#1}?(#2)\,}}
\newcommand{\return}[1]{\ensuremath{\mathbf{return}\,#1}}
\newcommand{\MethodDef}[2]{\ensuremath{\,\mathbf{mtd\ #1}(\mathbf{x})\,\{#2\}\,}}
\newcommand\futureLabel{\ensuremath{\mathit{future}}}
\newcommand\caller{\ensuremath{\mathit{caller}}}
\newcommand\MethodDefAsync[3]{\ensuremath{\,\mathbf{mtd\ #1}(\caller,\futureLabel,\mathbf{x})\,\{#2\sequence \return{#3}\}\,}}
\newcommand\identifier[1]{\ensuremath{\mathbf{#1}}}
\newcommand\invoke[1]{\ensuremath{\mathit{invoke}(#1)}}
\newcommand\completion[1]{\ensuremath{\mathit{comp}(#1)}}
\newcommand\assignment[2]{\ensuremath{\,\mathbf{#1}:=#2}}
\newcommand\upgradeTermSetKind[2]{\ensuremath{\,\mathbf{upgrade}^{#2}\!#1}}
\newcommand\newObjectOf[1]{\ensuremath{\,\mathbf{new}\,#1\,}}
\newcommand\classDefNameMethods[2]{\ensuremath{\,\mathbf{class\,#1}\,\{#2\}\,}}

\newcommand{\transition}[1]{\ensuremath{\xrightarrow{#1}}}

\newcounter{axiom}

\newcounter{axiomII}

\newcounter{examp}

\newcounter{mytablecounter}
\renewcommand{\themytablecounter}{\arabic{mytablecounter}}
\newcommand{\mytableheader}{\par\nopagebreak\vspace{0.01in}\noindent\par\nopagebreak}

               {\mytableheader\noindent%
               \setlength{\parindent}{4em}%
               \setlength{\parskip}{6ex}%
               \refstepcounter{mytablecounter}%
               \textbf{Table \nopagebreak[2]\themytablecounter{}: }}{\par}

\title{Dynamic Structural Operational Semantics\thanks{This work was partially supported by the project \href{http://its-wiki.no/wiki/DiversIoT:Home}{DiversIoT} --  Diversification for Resilient and Secure IoT-services, with number 270933/O70 part of the \href{http://www.forskningsradet.no/prognett-iktpluss/Home_page/1254002053513}{IKTPLUSS} program funded by the \href{http://www.forskningsradet.no}{Norwegian Research Council}.}}
\def\titlerunning{Dynamic Structural Operational Semantics (preliminary version)}

\author{
Christian Johansen
\institute{Institute for Technology Systems, University of Oslo}
\email{cristi@ifi.uio.no}
\and 
Olaf Owe
\institute{Department of Informatics, University of Oslo}
\email{olaf@ifi.uio.no}
}
\def\authorrunning{C. Johansen and O. Owe}

\begin{document}

\maketitle
\thispagestyle{empty}

\begin{abstract}
  We introduce Dynamic SOS as a framework for describing semantics of
  programming languages that include dynamic software upgrades, for
  upgrading software code during run-time. Dynamic SOS (DSOS) is built
  on top of the Modular SOS of P.\,Mosses, with an underlying category
  theory formalization. The idea of Dynamic SOS is to bring out the
  essential differences between dynamic upgrade constructs and program
  execution constructs. The important feature of Modular SOS (MSOS)
  that we exploit in DSOS is the sharp separation of the program
  execution code from the additional (data) structures needed at
  run-time.  In DSOS we aim to achieve the same modularity and
  decoupling for dynamic software upgrades. This is partly motivated
  by the long term goal of having machine-checkable proofs for general
  results like type safety.

We exemplify Dynamic SOS on two languages supporting dynamic software
upgrades, namely the \textsc{C}-like \proteusSrcLang, which supports
updating of variables, functions, records, or types at specific
program points, and \creol, which supports dynamic class upgrades in
the setting of concurrent objects. Existing type analyses for software
upgrades can be done on top of DSOS too, as we illustrate for
\proteusSrcLang.

As a side result we define of a general encapsulating construction on
Modular SOS useful in situations where a form of encapsulation of the
execution is needed. We use encapsulation in the \creol\ setting of
concurrent object-oriented programming with active objects and
asynchronous method calls.
\end{abstract} 
%
%
%

\tableofcontents

\section{Introduction}\label{sec_intro}

With renewed focus on software evolution
\cite{lehman80softEvolution,Mens2008
}, 
the interest in dynamic software upgrades has increased over
the past few years
\cite{MalabarbaPGBB00,DrossopoulouDDG02,BoyapatiLSMR03,JohnsenOS05dynamicClasses,AjmaniLS06,StoyleHBSN07mutatis,BiermanPN08}.
Approaches for dynamic upgrades are different in presentation and
formalization, making it difficult to compare or combine them,
especially since each of these approaches concentrates on some
particular programming language or paradigm. The work that we
undertake here is to extract the essentials of the operational
semantics for dynamic upgrading constructs independent of the
programming language or the kind of system paradigm.

Dynamic software upgrades provide mechanisms for upgrading a program at runtime, during its execution, by changing essential definitions used in executing the program, typically by adding or changing definitions of classes, interfaces, types, or methods, as well as modifying or resetting values of variables. Upgrades may be restricted, semantically or syntactically, so that they may only occur in certain states, called \textit{upgrade points}, where upgrading is meaningful or safe.
Dynamic upgrades allow a program to be corrected, improved, maintained or integrated with other programs, without stopping and restarting the execution. Dynamic upgrades are inherently different from normal  programming mechanisms because they are external to the program, using information that is not produced by the program, but is provided at runtime by an external entity or programmer.


Semantically, dynamic upgrades change \emph{static data structures}, i.e., the data structures established at the start of runtime such as class tables, function definitions and static typing information. This is in contrast to the semantics for normal programming constructs, which change the \emph{dynamic data structures} (also referred to as the {program state}), such as  the binding of values to program variables (the program store), heaps, message pools, or thread pools.

Thus at runtime we  distinguish between 
(i) the code being executed, (ii) the 
dynamic data structures, 
and (iii) the static data structures.
Standard operational semantics for programming languages is concerned with 
the runtime changes of the two former 
in the context of a given static data structure.
The complexity of the program state depends on the complexity of the language,
for instance, recursion requires a stack-based store.
Thus the operational semantics of a given code construct,
such as assignment,
may need to be reformulated when the language is enriched. 
 Modular SOS \cite{Mosses04modularSOS} solves this problem
by separating the structural layers of a program state.

In particular, Modular SOS (MSOS) promotes a sharp separation of the program code from the
  \emph{additional data structures}\footnote{Other works
    use the term \textit{auxiliary entities}, which we also use
    interchangeably throughout this paper to refer to the same
    concept.} that are manipulated by the semantics.
%
%
%
Moreover, complex features such as abrupt
termination and error propagation can be nicely handled by MSOS, as
well as combinations of big-step and small-step semantic styles. We
are not constrained in any way by building Dynamic SOS on MSOS. On the
contrary, MSOS is not binding the language designer to a notational
style. The notation can be the same as (or similar to) existing ones, as
soon as the concepts and style of MSOS and DSOS are adopted.  The
independence of notation is also seen in 
the work of Mosses and New \cite{MossesN09IMSOS}, which
presents new notational conventions called IMSOS, intended to be
attractive for the developers of programming languages.

We are interested in \textit{dynamic software updates} for imperative
languages such as the sequential C-like \proteusSrcLang\
\cite{StoyleHBSN07mutatis} and \textit{dynamic class upgrades} for
object-oriented languages such as the concurrent \creol\ language
\cite{JohnsenOS05dynamicClasses,JohnsenO07}. The nature of such
dynamic aspects is different from normal control flow and program
execution constructs of a language. Yet the interpretation of these
dynamic operations in the literature
\cite{StoyleHBSN07mutatis,JohnsenOS05dynamicClasses} is given using
the same style of structural operational semantics (SOS) as for the
other language constructs, often employing elaborate SOS definitions,
affecting the basic language elements as well as advanced ones.
Since the nature of dynamic upgrade constructs is different from
normal control flow and program execution constructs of a language, we
would like these differences to be apparent in the SOS descriptions.
For these reasons we develop Dynamic SOS (DSOS). 

The two chosen languages illustrate different kinds of dynamic
  updates.  \proteusSrcLang, which is the more low-level language,
  allows low-level state and code updates 
  as well as 
  control of the possible update points in the code.  \creol\ 
  is a high-level language for distributed systems supporting
  actor-like concurrent objects communicating by asynchronous methods
  calls and with support for high-level synchronization mechanisms
  including conditional process suspension.  
Each
  object has at most one active process, corresponding to a method
  activation, 
  while suspended processes are kept in a process queue of uncompleted
  method activations. Upgrades are done in a distributed manner; each
  object may upgrade itself at suspension or method completion. (As an
  aside, this
  allows program reasoning by means of dynamically updated class
  invariants, something which is a major concern in the \creol approach.)
Thus while the update points are programmer-defined in
\proteusSrcLang, they are predefined in \creol.

We show that DSOS can deal with both language settings in a uniform manner.


\textit{The contributions of this paper} are:
\begin{itemize}
\item We define a semantic framework for programming languages where
  dynamic software upgrades can be given semantics in a uniform
  manner, thus allowing for easier comparisons between 
different  upgrade machanisms.
\item We prove that
our DSOS framework is a conservative extension of the MSOS framework
that promotes modularity.

\item We show that typing aspects, commonly found in works on
  dynamic software upgrades, are readily doable on top of DSOS (like
  any other semantics), which we discuss in
  Section~\ref{sec_typing_proteus}.

\item In order to prove the adequacy of DSOS as a generalizing framework, we show how the semantics of two different languages with dynamic updates can be given in DSOS, i.e., we look at the popular \proteusSrcLang\ \cite{StoyleHBSN07mutatis} and at the more complex concurrent object-oriented \creol\ \cite{JohnsenO07}.

\item As a side result we introduce a method of \textit{encapsulation} on top of MSOS which we use in giving semantics to the object-oriented \creol. This is orthogonal and compatible with DSOS, and needed for exemplification purposes in the object-oriented setting.
When defining both DSOS and the encapsulation, we are concerned with respecting the principle of modularity, thus to be conservative extensions of MSOS. An MSOS treatment of object-oriented programming does not seam to appear elsewhere.
\end{itemize}

\subsection{An illustrative example}

We give a simple example to illustrate some aspects of dynamic software upgrades.
More complex examples can be found in e.g., \cite[Fig.\,3\,\&\,4]{StoyleHBSN07mutatis} from the Linux kernel, \cite[Sec.\,3]{JohnsenOS05dynamicClasses} for complex class upgrades, or \cite[Sec.\,3]{BiermanPN08}.

Consider a  class 
for keeping track of temperatures. The class
implements a simple interface
for setting and getting the (latest) temperature.
With  Java-like syntax it could look like

\begin{minipage}[t]{0.48\linewidth}
\begin{Creol}
interface Temp {
 void setTemp(int t)
 int getTemp() 
}
\end{Creol}
\end{minipage}%
\begin{minipage}[t]{0.48\linewidth}
\begin{Creol}
class TEMP implements Temp {
 int temp;
 void setTemp(int t){temp = t;}
 int  getTemp(){return temp;}
}
\end{Creol}
\end{minipage} 

Assume we would like to update a running system that uses this class
such that it can log the history of past \emph{temp} values and
is able to calculate the average temperature value.  We would like the
update to happen without restarting (and recompiling) the system.  In
\creol this is done by inserting into the message pool a runtime
upgrade message containing upgrade information
(using the keyword \kw{upgrade}),
which may redefine one or more classes
or add new classes and interfaces.
With high-level Java-like syntax the upgrade 
is given below:
\begin{Creol}
upgrade {
  interface TempStat extends Temp {int avgTemp()}

  class TEMP implements TempStat{
  int[] log = empty;

  void setTemp(int t){temp = t; log.append(t);}

  int avgTemp(){int avg=0; int i=0;
    for all x in log 
      {avg = avg + x; i= i+1;}
    return avg/i; }  //assuming non-empty log
}}
\end{Creol}
The upgrade introduces a new interface \texttt{TempStat} and a new
version of class \texttt{TEMP} augmented with a \texttt{log} variable, meant to store the sequence of temperature readings, as well as
a new method \texttt{avgTemp} for finding the average temperature. The actual logging is done in a
changed version of the original \texttt{setTemp} method.
The \texttt{getTemp} method is unchanged.
  (Note that
names are case sensitive, class names are written in upper case and
interface names start with an upper case character, while methods and
variables start with a lower case character.) 

The above example is presented in a syntax and style similar to a
\creol\ version.
In \creol, class upgrades are implemented in a distributed fashion
letting all the existing 
objects of class \texttt{TEMP} (or a subclass) make their upgrades
independently of each other \cite{JohnsenOS05dynamicClasses}.  An
update is performed when 
the current process in the object is suspended or completed.
Each upgraded object will
start to log temperature values, and will be able to respond to calls
to \texttt{avgTemp}.  Such calls may be generated by objects of
upgraded, or new, client classes.
Type safety is ensured by static checking of classes and of upgrades \cite{YuJO06dynamicTypes}.

In \proteus one may 
add a declaration of a new variable like for \texttt{log}, 
change the body of a function, like adding the \texttt{log.append} statement, 
add a new method, e.g., the \texttt{avgTemp}, 
and add calls to it, at predefined program points. 
The upgrades will be more fine-grained than in \creol, and to control when the updates are applied,  \proteus requires program update points to be pre-designated by the programmer, while for \creol the program update points are predefined by the concurrency model.

A challenge for the operational semantics is that such an upgrade as
above is changing the class and interface tables, as well as variable
and method bindings, in the middle of an execution.  In the \creol
case, upgrades are handled in the operational semantics by means of
message passing, 
by sending special upgrade messages (like the \emph{upgrade}
  definition above).
However, a complicating factor of the operational semantics is that
\creol level messages (reflecting method invocations and returns) and
upgrade level messages are using the same underlying message passing
mechanism.

\subsection{Dynamic SOS}

We are taking a \textit{modular} approach to SOS, following the work
of Mosses \cite{Mosses04modularSOS}, thus building on \textit{Modular
  SOS} (MSOS).  This formalism uses notions of category theory, on which our work depends.
 Dynamic SOS is intended as a framework for studying semantics
of dynamic upgrade programming constructs, 
and thus existing works on dynamic upgrades should be naturally captured; 
we exemplify DSOS on the
\textit{dynamic software updates} of the language \proteusSrcLang\
\cite{StoyleHBSN07mutatis} and on the \textit{dynamic class upgrades}
of the concurrent object-oriented language \creol\ 
\cite{JohnsenOS05dynamicClasses,JohnsenO07}.  Since much of the literature on software updates focuses on type systems and type safety, and since
their results also hold over  Dynamic SOS, here we concentrate mainly on
the semantic aspect, and only briefly discuss typing aspects in
Section~\ref{sec_typing_proteus}.

One observation that 
we want to emphasis with DSOS
is that \textit{upgrade points} must
be identified and marked accordingly in the program code. The marking
should be done with special upgrade programming constructs. Here we
are influenced by the work on \proteusSrcLang\
\cite{StoyleHBSN07mutatis} (which is also taken up in \upgradej\
\cite{BiermanPN08} and the multi-threaded \stump\
\cite{NeamtiuH09}). Opposed to a single marker as in \proteusSrcLang,
one could use multiple markers. This would allow also for incremental
upgrades. The purpose of identifying and marking such upgrade points
is to ensure type safety after upgrades. The analysis techniques of
\cite{StoyleHBSN07mutatis} for safety after upgrades can be used over
DSOS as well.
Upgrade markers can be placed by a programmer or automatically by static analysis techniques, as in \cite{StoyleHBSN07mutatis}.

A second observation is that compared to the normal flow of control
and change of additional data that the execution of the program does,
we view a dynamic upgrade as a contextual \textit{jump} to a possibly
completely different static structure (i.e., data content). This, in
consequence, can completely alter the execution of the
program. Moreover, these jumps are strongly knit to the upgrade
information, which is regarded as outside the scope of the executing
program, being externally provided.  If normal program execution
changes to the static structures are captured by the
\textit{morphisms} in the MSOS style, the jumps will be captured in
DSOS using \textit{endo\-functors}, a concept of higher abstraction,
which are still seen as morphisms in an appropriate category,
as we explain later on.

When seeing new frameworks, like DSOS, one may wonder about their purpose, and especially whether the same could be done with what already exists.
First, we see as a contribution any good attempt to unify seemingly disparate concepts, to allow easier comparisons and future developments of similar concepts. This is particularly so with the various dynamic software upgrade constructs out there, giving us one motivation for developing DSOS by identifying common features, and lifting these to a more abstract level of a framework.
Second, one may ask whether the DSOS mechanisms can be captured by an encoding solely within the MSOS framework.
The authors could not find a reasonable answer to this, and thus leave it for future work. Nevertheless, even if dynamic upgrade concepts could be encoded in MSOS, one then needs to study how natural would this encoding be, and whether it would help or not programming language designers.
More specific discussions on these lines are done in the concluding Section~\ref{sec_conclusion}.

\subsection{Modular semantics for concurrent object-orientated languages}

A second contribution of this paper is to enhance the theory of
Modular SOS with a general notion of \emph{encapsulation} that helps give
semantics when a form of encapsulation of the
execution is needed, such as in the setting of concurrent and distributed systems.
The concurrency model that we treat here,
and which is useful in an object-oriented setting, 
is that of the \textit{Actor
  model} \cite{HewittBS73actorModel,AghaMST97actorFoundation} where
each concurrent entity is autonomous, thought as running on one
dedicated machine or processor. Therefore, the auxiliary data
structures that the standard SOS employs are also localized to each
actor. We capture this localization mechanism in a general manner,
yet staying in the framework of MSOS, by making a construction on the
category theory of MSOS, which we call the \textit{encapsulating
  construction}, and show it to be in agreement with the other
category notions of MSOS.
This is worked out  in the  setting of object-oriented
programming with concurrent objects of \creol. Object-orientation has not been
treated before in the MSOS style. 
However, 
concurrent ML was treated in \cite{mosses99concurrentML}.

Executable semantics of programming languages prototypes has been advocated by the \creol\ since early papers \cite{JohnsenO07,JO04acm} where SOS-style of semantics were implemented in the rewriting logic of the Maude system \cite{CLAVEL2002Maude}. Similar goals of automating programming languages semantics are shared by other works as well, e.g., either executing and simulating it or giving it as input to a proof assistant \cite{Sewell2010ott,Klein2012RYR}. The MSOS style of semantics can also be implemented in the rewriting logic of Maude \cite{MaudeMSOS}.

\subsection{Structure of the paper}\hspace{1ex} 
We first give a short listing of some simple notions of category
theory that will be used throughout the paper and then introduce
Modular SOS in Section~\ref{sec_MSOS}. We then exemplify, in
Section~\ref{sec_MSOS_proteus}, MSOS on constructs found in the
\proteusSrcLang\ language, following a modular style of giving
semantics to one programming construct at a time. 
In the end, the
language and its semantics are formed by summing up the needed
syntactic constructs with their respective MSOS semantic elements and
rules. 
%
In Section~\ref{sec_DSOS}
we develop the Dynamic SOS theory, our main contribution.
Both 
\proteusSrcLang\ and \creol\ 
have dynamic upgrading constructs which are given semantics in Sections~\ref{subsec_DSOS_proteus} and \ref{subsec_ex_DSOS_creol}, respectively.
We discuss in Section~\ref{sec_typing_proteus} how typing aspects from
standard papers on dynamic upgrades can be done over DSOS as well, and
look particularly at \proteusSrcLang.
In Section~\ref{sec_encapsulation} we introduce the encapsulating
construction and use it in Section~\ref{subsec_MSOS_creol} to give
modular semantics to concurrent object-oriented constructs found in
the \creol\ language. 
We conclude and discuss possible applications and continuations of this work in Section~\ref{sec_conclusion}.

\section{Modular Structural Operational Semantics}\label{sec_MSOS}

The usual structure of papers on programming languages would include a
section that introduces the \textit{syntax} of the language studied,
which would then be followed by a section describing the semantics.
This is contrary to how DSOS and MSOS propose to develop (semantics of)
programming languages. In DSOS we give semantics to a single
programming construct, independently of any other constructs (as one
can later see through the examples that we give). To define a
programming language one puts together the syntactic constructs and
the respective semantic rules. Such an approach is particularly
appealing when developing a programming language assisted by a
theorem prover (e.g., \cite{pierce12SOS_coq}).
A main goal of the modular approach is to ensure that once the
semantics has been given to one programming construct, it does not need
to be changed in the future, when adding new programming constructs.
This will be illustrated 
throughout our presentation.

Moreover, it is easy to work within different notational
conventions. Translations between these notations are 
 possible because of the common underlying theory provided by the MSOS and its category theory foundations. Nevertheless, these categorical foundations are transparent to the one giving semantics to programming languages. Standard notational conventions can be adopted for MSOS, but the methodology changes to a modular way of thinking about the semantics. The independence of notation can be seen in \cite{MossesN09IMSOS}, which presents new notation conventions called IMSOS, intended to be more attractive to the designers of programming languages.

 We recall briefly some standard technical notions that will be used
 throughout this paper. 
Our notation stays close  to that of
\cite{Mosses04modularSOS} for the MSOS related notions and to that of
\cite{pierce91categoryBook} for other notions of category theory.

\begin{definition}[category]\label{def_category_app}
  A \emph{category} (which we denote by capital letters of the form
  $\cat{A}$) consists of a set of \emph{objects} (which we denote by
  $\objects{\cat{A}}$ with usual representatives $o,o',o_{i}$) and a
  set of \emph{morphisms}, also called \emph{arrows}, between two
  objects (which we denote by $\morphisms{\cat{A}}$ with usual
  representatives $\alpha,\beta$, possibly indexed). A morphism has a
  \emph{source} object and a \emph{target} object which we denote by
  $\source{\alpha}$ and $\target{\alpha}$. A category is required  
(i) to have \emph{identity} morphisms $id_{o}$ for each object $o$,
  satisfying an identity law for each morphism with source or target
  in that object; and 
 (ii) composition of any two morphisms
  $\alpha$ and $\beta$, with $\target{\alpha}=\source{\beta}$, exists
  (denoted $\beta\composition\alpha$, or just $\alpha\beta$, as in
  computer science) and is associative.
\end{definition}

\begin{definition}[functors]\label{def_functor}
Consider two arbitrary categories \cat{A} and \cat{B}. A \emph{functor} $F:\cat{A}\rightarrow\cat{B}$ is defined as a map that takes each object of $\objects{\cat{A}}$ to some object of $\objects{\cat{B}}$, and takes each morphism $\alpha\in\morphisms{\cat{A}}$ to some morphism $\beta\in\morphisms{\cat{B}}$ s.t.\ $o\transition{\alpha}o'$ is associated to some $F(o)\transition{\beta}F(o')$, and the following hold:
\[
F(id_{o})=id_{F(o)} \mbox{\hspace{3ex}and\hspace{3ex}} F(\alpha\beta)=F(\alpha)F(\beta).
\]
A functor $F:\cat{A}\rightarrow\cat{A}$ is called an \emph{endo\-functor} applied to $\cat{A}$ (or on \cat{A}).
Define $\mathit{End}(\cat{A})$ the \emph{category of endo\-functors on \cat{A}}, having \cat{A} as the single object and endo\-functors on \cat{A} as morphisms.
\end{definition}

Modular SOS generates \textit{arrow-labelled transition systems}, cf.\ \cite{Mosses99foundationsMSOS}, where the transitions are labelled with morphisms (arrows) from a category.

\begin{definition}[ALTS]\label{def_arrowTS}
An \emph{arrow-labelled transition system}
$(\Gamma,\morphisms{\cat{A}},\transition{})$ is formed
by 
a set of
\emph{states} $t_{i}\!\in\!\Gamma$,
including an \emph{initial state} $t_0$, and transitions  $\transition{}\subseteq\!\Gamma\!\times\!\morphisms{\cat{A}}\!\times\!\Gamma$, labelled by morphisms $\alpha\!\in\!\morphisms{\cat{A}}$ from a category $\cat{A}$.
A \emph{com\-putation} in an ALTS is a sequence $t_{0}\!\transition{\alpha_{0}}\!t_{1}\!\transition{\alpha_{1}}\!t_{2}\dots$ s.t.\ for any $t_{i}\!\transition{\alpha_{i}}t_{i+1}\!\transition{\alpha_{i+1}}t_{i+2}$ the two morphisms are composable in \cat{A} as $\alpha_{i+1}\circ\alpha_{i}\in\!\morphisms{\cat{A}}$.
\end{definition}

\begin{notation}
Since in an ALTS transitions $\transition{\alpha}$ are labelled with morphisms from $\cat{A}$, we also have a grip on the underlying objects involved in the transition, i.e., $\source{\alpha}$ and $\target{\alpha}$. 
When the source and target objects of the morphism $\alpha$ are needed we make them explicit on the transition as $\transition{\{\source{\alpha},\target{\alpha}\}}$.
\end{notation}

One goal with ALTS and MSOS is to have as states \textit{only program terms}, without the additional semantic data that an executing program may use, like stores or heaps. 
The additional data and the way the program manipulates it is captured by the morphisms which are labelling the transitions of the ALTS. 
This goal is related to e.g.: 
\begin{enumerate}
\item typing systems where the program syntax alone is under analysis; 

\item Hoare logic where Hoare rules are defined for program terms only (with the pre- and post-conditions being the ones talking about the stores/heaps); 

\item process algebras with process terms as the states and their observable behaviour as labels on transitions.
\end{enumerate}

When giving semantics to programming languages we establish an initial multi-sorted signature defining the programming constructs of interest. This signature may be enriched upon future developments of the language with new programming constructs. 
The closed program terms built over this signature constitute the configurations of the arrow-labelled transition systems. 
Any additional structure/data (like heaps or stores) needed when giving semantics to these constructs, are objects in special categories from which we take their morphisms as transition labels.

\begin{definition}[basic label categories]\label{def_labcat}
The following three kinds of categories, called \emph{basic label categories}, are used to build more complex label categories:

\begin{itemize}
\item \textbf{discrete category:} A \emph{discrete category} is a category which has only identity morphisms. No other morphisms are allowed.

\item \textbf{pairs category:} A \emph{pairs category} is a category which has one unique morphism between every two objects (i.e., in each direction).

\item \textbf{monoid category:} A \emph{monoid category} is a category that has a single object and the morphisms are elements from some predefined set $\mathit{Act}$.
\end{itemize}
\end{definition}

Intuitively, discrete categories correspond to additional information
that is of a read-only type, like read-only variables.
Pairs categories correspond to additional data of a read/write type, like stores. Each store appears as one object in the category. The morphisms between two stores represent how a store may be modified by the program when executed. 
We take a general view where a program may change a store in radical ways, therefore, we have morphisms between every two stores.
Monoid categories correspond to write-only type of data, like observable information emitted during the execution of the program, or messages sent between communicating processes.

\begin{myexample}\label{example_category_kinds}
  To build a monoid category we pick an underlying set of actions (or
  events) which will make a monoid of strings over this alphabet, with
  the empty string as the identity morphism.  We can build a discrete
  or a pairs category by picking some underlying set of objects.  One
  standard example of a pairs category \cat{S} has as objects stores:
  $\objects{\cat{S}}=\mathit{IdVar}\rightharpoonup\mathit{Val}$,
  i.e. all partial functions from some set $\mathit{IdVar}$ of
  variable identifiers to some set $\mathit{Val}$ of values.  
\end{myexample}

When several additional data are needed to define the semantics, we use
{\it{complex label categories}} obtained by making product of basic
label categories. 
One may use as many data components as needed to get a natural view of the
semantics for each programming construct.

An implementation may choose to put several data
structures together, if no clashes can appear.
Complex labels are built using the following construction, which attaches an index to each label component. This will offer the possibility to uniquely identify each component from a complex label using the associated index.
This also provides a modular way of extending the label categories.

\begin{definition}[label transformers]\label{def_labtrans}
  Let $\indexE$ be a countable set of indexes, \cat{B} a basic label
  category, and $\cat{A}\!=\!\mathop{\prod}_{j\in J\subset
    \indexE}\cat{A}_{j}$ a product category which is the trivial
  category\footnote{The trivial category has a single object and only the identity morphism for it.} 
\trivcat\ when $J\!=\!\emptyset$. 
A \emph{label transformer} $\labtrans{i}{\cat{B}}$, with
  $i\!\in\!\indexE\setminus\!J$, maps \cat{A} to the product category
  $\cat{A}\prodcat\cat{B} = \labtrans{i}{\cat{B}}(\cat{A})$, and associates a partial operation
\[
\getcat:\morphisms{\cat{A}\prodcat\cat{B}}\times\indexE\rightarrow (\cup_{j}\morphisms{\cat{A}_{j}})\cup\morphisms{\cat{B}}
\]
which for
  each composed morphism of the new $\cat{A}\prodcat\cat{B}$
  associates a morphism in one of the component categories of the
  product, as follows:
\[
\getcat((\alpha_{\cat{A}},\beta_{\cat{B}}),k)=\left\{
\begin{array}{ll}
\beta_{\cat{B}}, & \mbox{if } i=k\\
\getcat(\alpha_{\cat{A}},k), & \mbox{otherwise}.\\
\end{array}\right.
\]
\end{definition}

\begin{notation} 
For a composed morphism $\alpha$ of a product category obtained using the label transformer we may denote the get operation using the dot-notation (well established in object-oriented languages) to refer to the respective component morphism; i.e., $\alpha.i$ for $\getcat(\alpha,i)$, with $i$ being one of the indexes used to construct the product category.
Since $\alpha.i$ is a morphism in a basic label category, we may also refer to its source and target objects (when relevant, like in the case of discrete or pairs categories) as $\source{\alpha.i}$ respectively $\target{\alpha.i}$.
\end{notation}

Now we proceed to define how operational rules look like in this setting.

\begin{definition}[program terms]\label{def_progTerms}
A \emph{multi-sorted signature} $\Sigma$ is a set of function symbols, together with an \emph{arity} mapping $\arity{}$ that assigns a natural number to each function symbol, and a family of \emph{sorts} $S_{i}$. Each function symbol has a sort definition which specifies what sorts correspond to its inputs and output. A function of arity zero is called a \emph{constant}.
The set of \emph{terms}
over a signature $\Sigma$ and a set \metaVariables\ of sorted meta-variables is denoted $\sortedterms{}{\Sigma,\metaVariables}$ and is defined as follows (we often omit the set \metaVariables\ for readability) :
\begin{itemize}
\item any meta-variable is a term;
\item a function application $f(t_{1},\dots,t_{\arity{f}})$
for some function symbol $f$ and set of terms
$t_{1},\dots,t_{\arity{f}}$, of the right sort, is a term. 
\end{itemize}
\end{definition}




%
\begin{definition}[rules]\label{def_rules}
We call $t\transition{\alpha}t'$ a \emph{transition literal} (or transition schema), with $t,t'$ program terms, possibly containing meta-variables (i.e., these are program schemes). A transition schema is \emph{closed} iff $t,t'$ are, i.e., do not contain meta-variables. 
The $\alpha$ is a specification of a set of morphisms allowed as labels of this transition schema (see Notation~\ref{notation_morphTrans}).
A transition \emph{derivation rule} is of the form $H/l$ with $H$ a set of transition literals, called the \emph{premises}, and $l$ is a single transition literal, called the \emph{conclusion}.
\end{definition}

When side-conditions (e.g., equations, set memberships, definedness assertions) are needed in a rule, we write these on top of the derivation line, together with the premises, since they can easily be distinguished from transition literals.
Negations of side-conditions can also be used.

\begin{definition}[generated ALTS]\label{def_generatedALTS}
The semantics of a program $P$ is defined as the \emph{generated arrow-labelled transition system} that has as states closed program terms, as initial state the program $P$, and as transitions all the closed transitions generated by exhaustively instantiating the derivation rules.
\end{definition}

\begin{notation}[morphisms on transitions]\label{notation_morphTrans}
When writing literals we use the following notation for the labels. We
write
$t\transition{\{\source{\alpha.i}\restlabel\target{\alpha.i}\}}t'$ to
mean that the morphism $\alpha$ is a tuple where the label component
indexed by $i$ is the one given on the transition, and all other
components are the identity morphism, symbolized by the \textit{three
  dots}. We write sources of morphisms to the left of the three dots,
and targets to the right. In one transition we may refer to several
components, e.g.:
$t\transition{\{\source{\alpha.i},\alpha.j\restlabel\target{\alpha.i}\}}t'$. In
this example the $j$ index is associated with a discrete category, and
therefore we do not write the target of it on the right because it is
understood as being the same. Moreover, because of the right/left
convention we omit the superscripts. An even more terse notation may
simply drop all references to $\alpha$ and keep only the indexes, thus
the last example becomes $t\transition{\{i=o,j=h\restlabel
  i=o'\}}t'$. The objects $o,o'$ may be stores, and thus the
transition says that the store $o$ is changed to 
$o'$, whereas the
component $j$ may only be inspected. 
\end{notation}

The goal of \textit{modularity} is to have rules defined once and for
all, meaning that when a new programming construct is added and the
new rules for it need 
to refer to new auxiliary semantic 
entities, i.e., to enlarge the old label category, then the old rules
need not be changed.  This is made precise by the essential result of
\cite[Prop.1]{Mosses99foundationsMSOS}.

Intuitively, this result says that any transition defined using the old rule system, i.e., labelled with some $\alpha$ from some category \cat{A}, is found in the new arrow-labelled transition system, over a new category \labtrans{i}{\cat{B}}(\cat{A}), using an embedding functor which just attaches an identity morphism to the old morphism, i.e., $(\alpha,id_{b})$, for the current object $b\in\objects{\cat{B}}$. Moreover, for any transition defined in terms of the new composed labels from \cat{A}\prodcat\cat{B}, if it comes from the old rules only then the projection from \cat{A}\prodcat\cat{B}\ to \cat{A}\ gives an old label morphism by forgetting the identity morphism on \cat{B}. This is the case because the old transition refers only to components in \cat{A}, where the dots notation makes all other components contribute only with the identity morphism.

\begin{theorem}[{\cite[Proposition 1]{Mosses99foundationsMSOS}}]
Let \cat{A} be a category constructed using the label transformers $\labtrans{j}{\cat{B}_{j}}$ for some basic label categories $\cat{B}_{j}$ of the three kinds defined before, with $j\in J \subset \mathit{Index}$. Consider a set of rules $R$ which specifies an ALTS over \cat{A}, where the rules in $R$ refer to only indexes from $J$.
Let the category $\cat{A}'=\labtrans{i}{\cat{B}_{i}}(\cat{A})$, where $i\not\in J$, and let $\transition{}'$ be the transition relation  specified by the same set of rules $R$ but having labels from $\cat{A}'$.
For each computation $\transition{\alpha}\transition{\beta}\dots$ specified by $R$ over \cat{A}, we have a corresponding computation $\stackrel{\alpha'}{\longrightarrow'}\stackrel{\beta'}{\longrightarrow'}\dots$ over $\cat{A}'$, and vice versa.
\end{theorem}

\begin{proof}[sketch]
This result is a consequence of \cite[Proposition 1]{Mosses99foundationsMSOS} and is explicitly stated in the corresponding technical report \cite[Corollary 1]{Mosses99foundationsMSOStr}.
The label transformer $\labtrans{i}{\cat{B}_{i}}$ forms a projection functor from $\cat{A}\prodcat\cat{B}_{i}$. This functor is used to get the reverse direction of the statement, by forgetting the structure of $\cat{B}_{i}$. This is possible because the rules in $R$ do not refer to this index $i$, hence to morphisms in $\cat{B}_{i}$, which means these are just the identity morphisms. The label transformer also forms a family of embedding functors from $\cat{A}$ into $\cat{A}\prodcat\cat{B}_{i}$ (for each object of $\cat{B}_{i}$). These functors are used to obtain the forward direction of the statement. Depending on the current object of $\cat{B}_{i}$ we use the corresponding embedding functor to add to the label specified by the rules $R$ on $\cat{A}$ an identity functor on $\cat{B}_{i}$, thus obtaining a corresponding transition with label morphism from $\cat{A}\prodcat\cat{B}_{i}$.
\end{proof}

\section{Exemplifying MSOS for the language \proteusSrcLang}\label{sec_MSOS_proteus}

Normally, for exemplifying
how the theory of Modular SOS is applied, it would be advised to use a
minimal set of programming constructs.  However, we want the
theory to appeal to practitioners that develop programming languages.
We therefore consider various common programming constructs, 
without being concerned about redundancy.
Moreover, the constructs that we treat in this section will sum
up to 
the programming language \proteusSrcLang\
\cite{StoyleHBSN05popl}.
We use though slightly different constructs, closer to those found in the \creol\ language \cite{JohnsenOS05dynamicClasses}, which we treat in
Section~\ref{sec_ex_DSOS_creol}.
Thus we focus on imperative contructs with static typing
and static variable binding.
To simplify the presentation we assume program variables have distict
names
(which could be achieved by adding the declaration level to  
variables during static analysis). 
Our style of giving semantics in this section is incremental, one construct at a time.
While giving the semantics we deliberately want to be free from any specific notation convention; i.e., we want to convey the concepts of modular semantics, and not adhere to a particular established way of giving SOS semantics to programming languages. We afford to do this because of the modular framework \cite{MossesN09IMSOS,ChurchillMST14PlanCompsTAOSD}. 

This section can be skipped by a reader knowledgeable of MSOS, though, if some later notation seams unclear one can come back to this section for clarifications.

Throughout the paper we work with what is sometimes called \emph{value-added syntax}, where the values that program constructs work with are included in the language syntax as constant symbols. Denote these generally as $v\in\mathit{Val}$, with $n\in\mathbb{N}\subseteq\mathit{Val}$ and $b\in\{\mathbf{true}, \mathbf{false}\}\subseteq\mathit{Val}$.
The $\nilValue\in\mathit{Val}$ is seen as a special value that statements take when finished executing.
The values are considered to have sort \emph{Expressions}, denoted usually by $e\in E$.

\subsection{No labels for sequential composition}\label{example_noLabels}
Consider a \emph{sorted} signature $\SigmaExec_{\ref{example_noLabels}}$ consisting of the following programming constructs, forming a single sort \textit{Statement}:
\[
s\ \ ::=\ \ \skipStatement \mid s\sequence s
\]
where $\skipStatement$ is a constant, standing for the program that
does nothing, and  $\_\sequence\_$  is
a binary function symbol, standing for \textit{sequential composition}.


\begin{remark}[numbering the signatures]
We use a subscript to number the different signatures that we construct. We use as number the reference of the respective subsection where the signature is defined. 
\end{remark}

We define the following transition rules:
\begin{prooftree}
\AxiomC{\phantom{$\ \transition{\{\restlabel\}} \ $}}
\UnaryInfC{$\skipStatement \transition{\{\restlabel\}} \nilValue$}
\DisplayProof\hspace{3ex}
\AxiomC{$s_{1} \transition{X} s_{1}'$}
\UnaryInfC{$s_{1}\sequence s_{2} \transition{X} s_{1}'\sequence s_{2}$}
\DisplayProof\hspace{3ex}
\AxiomC{\phantom{$\ \transition{\{\restlabel\}} \ $}}
\end{prooftree}
We assume that  $\_\sequence\_$ is an associative constructor
with  $\nilValue$ as left and right identity element,
and we assume pattern matching of the transition rules
 modulo associativity and identity (as for instance supported by 
rewriting logic/Maude\footnote{The Maude System: \url{http://maude.cs.uiuc.edu/}} \cite{CLAVEL2002Maude,MaudeMSOS}).
%
The special label variable $X$ stands for \textit{any morphism}, and the label $\{\restlabel\}$ stands for  \textit{any identity morphism}. 
These rules do not specify label categories because any category can be used. This means that no additional data is needed by the respective two programming constructs.
Moreover, the identify morphisms capture naturally the notion of \emph{unobservable} transitions since they just ``copy'' the data represented by the objects.

The second rule has one premise, and assumes nothing about the morphism of the transition; it only says that the label is carried along from the statement $s_{1}$ to the whole sequence statement $s_{1}\sequence s_{2}$. 
%
The first rule is an \textit{axiom} because it contains no premises, and says that the \skipStatement\ program reduces to the \textit{value} \nilValue\ by the identity morphism on the current object in the current category of labels, whichever this may be. 
We also consider to have the standard arithmetic and Boolean operators which take expressions and return expressions. 
(See technical report \cite{techRep12Dec} for a detailed example.)

\subsection{Read-only label categories and a let construct}\label{example_variables}
We add a set of variable identifiers as constant symbols, and denote these by $\mathbf{x}\in\mathit{IdVar}$. Variable identifiers have sort \emph{Expressions}. 
We also include a $\mathbf{let}$ construct usually found in functional languages.
Let these make a signature $\SigmaExec_{\ref{example_variables}}$, which can be added to any other signature.
\[
e\ \ ::=\ \ \mathbf{x} \ \mid 
\letStatement{\mathbf{var\ x}:=e'}{e} \ \mid \dots
\]

%
The interpretation of variable identifiers is given wrt.\ \textit{an
  additional data structure called store}, which keeps track of the
values associated to each variable identifier.
In consequence, we define a label category \cat{S}, having as objects $\objects{\cat{S}}=\mathit{IdVar}\rightharpoonup\mathit{Val}$ the set of all partial functions from variable identifiers to values, denoting stores.
Define \cat{S} as a \textit{discrete category}, i.e., only with identity morphisms, since in the case of variable identifiers alone, the store is intended only to be inspected by the program. The label category to be used for defining the transitions is formed by applying the label transformer $\labtrans{S}{\cat{S}}$ to any category of labels, depending on the already chosen programming constructs and transition rules; in our case to the trivial category, since no specific label components were used until now.
Instead of using natural numbers as indexes we use symbols. However, notational decisions are relative to the user, and our notation choices from this paper can safely be overridden.


The transition rule corresponding to the variable identifiers is:
\begin{prooftree}
\AxiomC{$\rho(\mathbf{x})=v$}
\UnaryInfC{$\mathbf{x} \transition{\{S=\rho\restlabel\}} v$}
\end{prooftree}

\vspace{1ex}The rule defines a transition between terms $\identifier{x}$ and $v$, labelled with a morphism satisfying the condition that the label component with index $S$ has as source an object $\rho\in\objects{\cat{S}}$ that maps the variable identifier to the value $v$. 
Because the category \cat{S} is discrete, we do not specify the target object explicitly since it is the same as the source object specified on the label, i.e., an identity morphism is used.
Any other possible label components, if and when they exist, contribute with an identity morphism (symbolized by the three dots). 
In consequence, since all morphism are identity, this transition is unobservable.
Henceforth, whenever in a rule we mention only the source of a morphism component it means that the target is the same, i.e., we specify only some particular identity morphisms.
%
%
Note that the rules from Section~\ref{example_noLabels} are unaffected by the fact that we have changed the label category. Neither will future rules be affected.

The semantic rules for $\mathbf{let}$ are given in a small-step style
using textual
substitution $[v/x]$ as in \cite{StoyleHBSN07mutatis} or \cite[Sec.4.1]{mosses99concurrentML}, assuming that all variable names are distinct (i.e., an application of Barendregt's variable convention \cite[p.26]{Barendregt81bookLambda}).







\begin{prooftree}
\AxiomC{$e'\transition{X} e''$}
\hspace{-2ex}\UnaryInfC{$\letStatement{\mathbf{var}\ \mathbf{x}:=e'}{e} \transition{X} \letStatement{\mathbf{var}\ \mathbf{x}:=e''}{e}$}
\DisplayProof\hspace{0ex}
\AxiomC{\phantom{$\ \transition{U} \ $}}
\UnaryInfC{$\letStatement{\mathbf{var}\ \mathbf{x}:=v}{e} \transition{\{\restlabel\}} e[v/\mathbf{x}]$}
\end{prooftree}
\subsection{Changing label categories from read-only to read/write for assignments}\label{example_assignment}
Having variable identifiers we may add \emph{assignment} statements and \emph{variable declarations} as $\SigmaExec_{\ref{example_assignment}}$, which would include $\SigmaExec_{\ref{example_variables}}$. 
%
%
\[
d\ ::=\ \mathbf{var}\ \mathbf{x}:=e \mid \dots 
\hspace{3ex}
s\ ::=\ \mathbf{x}:=e \mid d \mid \dots
\]

Both assignments and declarations (which are a subsort of statements) allow the program to \textit{change} the store data structure that we used before for evaluating variable identifiers. 
Therefore, here we need \cat{S} to be a \textit{pairs category} so to capture that a program can also change a store, besides inspecting it.
Important in Modular SOS is that rules which use read-only discrete categories are not affected if we change these label components to be read/write pairs categories (with the same objects). Indeed, the syntax used in the rules refers only to the source objects of the morphisms. 
In consequence, the new label category is made using the label transformers exactly as before, only that when adding the component with the index $S$ we add \cat{S} as a pairs category. All the rules from before use the identity morphisms on the objects. The new rules that we add use proper pair morphism, i.e., referring to both the source and the target stores of the morphism.
%
%
%
\begin{prooftree}
\AxiomC{$e\transition{X}e'$}
\UnaryInfC{$\mathbf{var}\ \mathbf{x}:=e \transition{X} \mathbf{var}\ \mathbf{x}:=e'$}
\DisplayProof\hspace{3ex}
\AxiomC{$\mathbf{x}\not\in\rho$}
\UnaryInfC{$\mathbf{var}\ \mathbf{x}:=v \transition{\{S=\rho\restlabel S=\rho[\mathbf{x}\mapsto v]\}} \nilValue$}
\end{prooftree}
The premise of the second rule can be ensured by
the typing system, and thus could be removed. This is even desired when
we want the rules to be in a standard rule format
\cite{aceto01SOS_handbook,Mousavi07SOSsurvey}. However, rule formats
for DSOS are deferred to future work, discussed in
Section~\ref{subsec_further_work}, where one would need to look at
more recent works on formats for data
\cite{Mousavi06SOSdata,Mousavi13algData} and for MSOS
\cite{ChurchillMM13MSOSformats}.
The rules for assignment are similar.
\begin{prooftree}
\AxiomC{$e\transition{X}e'$}
\UnaryInfC{$\mathbf{x}:=e \transition{X}  \mathbf{x}:=e'$}
\DisplayProof\hspace{3ex}
\AxiomC{$\mathbf{x}\in\rho$}
\UnaryInfC{$\mathbf{x}:=v \transition{\{S=\rho\restlabel S=\rho[\mathbf{x}\mapsto v]\}} \nilValue$}
\end{prooftree}
Again, the premise of the second rule can be guaranteed by type checking
(assuming static binding).
%

\subsection{Functions}\label{example_functions}
Consider \emph{function identifiers} as constants denoted by  $\mathbf{f}\in\idFunc$, and \emph{function definitions} and \emph{function applications}, in the signature $\SigmaExec_{\ref{example_functions}}$ below. This may be added to any signature that includes variable identifiers, like $\SigmaExec_{\ref{example_assignment}}$.

\[
d\ \ ::=\ \ \funcdecl{f}{\mathbf{x}}{s} \mid \dots 
\hspace{6ex}
s\ \ ::=\ \ \mathbf{f}\,e \mid \dots 
\]
Function declarations are stored in a new label component which is a pairs category\footnote{Normally, the program at runtime just inspects the function definitions, therefore we could consider using a read-only, discrete, category label component. However, we are using above function definitions as programming constructs. Their semantics is exactly to change the stored definitions of functions.} containing objects which associate function identifiers to lambda terms. Denote this category by \cat{F} and its objects as $\rho_{f}\in\objects{\cat{F}}$. Add this as a label component using the label transformer 
$\labtrans{F}{\cat{F}}\circ\labtrans{S}{\cat{S}}$. Since variable identifiers are needed, the stores component is added as well.

Another semantics, like that of \cite{StoyleHBSN07mutatis}, may want to consider these two as a single store-like data structure. 
In this paper we prefer to use disjoint structures when possible.
At an implementation stage one could merge these two kinds of stores into one, and take care of differentiating the variable identifiers from the function identifiers correctly.

The transition rules below are as in \proteusSrcLang, using a functional languages style. 
We are using again the notation $s[v/\mathbf{x}]$ for 
substitution 
of all occurrences of the variable in the statement body of the function.
This is typical for reduction semantics, as in \cite{StoyleHBSN07mutatis}; however we could also use evaluation contexts, e.g., as done in \cite{Mosses04modularSOS}. We exemplify the use of evaluation contexts in Section~\ref{example_threads} for the semantics of threads.
\begin{prooftree}
\AxiomC{$e\transition{X}e'$}
\UnaryInfC{$\mathbf{f}\,e \transition{X} \mathbf{f}\,e'$}
\DisplayProof\hspace{1ex}
\AxiomC{$\rho_{f}(\mathbf{f})=\lambda(\mathbf{x}).s$}
\UnaryInfC{$\mathbf{f}\,v \transition{\{F=\rho_{f}\restlabel\}} s[v/\mathbf{x}]$}
\DisplayProof\hspace{1ex}
\AxiomC{\phantom{$\mathbf{f}\not\in\rho_{f}$}}
\UnaryInfC{$\funcdecl{f}{\mathbf{x}}{s}
  \transition{\{F=\rho_{f}\restlabel F=\rho_{f}[\mathbf{f}\mapsto
    \lambda(\mathbf{x}).s]\}} \nilValue$}
\end{prooftree}
%

\subsection{Records}\label{example_records}
We add to $\SigmaExec_{\ref{example_functions}}$ a set of \emph{record names} as constants $\mathbf{r}\in\idRecord$ and a set of \emph{record labels} as constants $\mathbf{l}\in\idRecordLab$, together with two language constructs for \emph{record definition} and \emph{record projection}, thus making $\SigmaExec_{\ref{example_records}}$:
\[
d\ \ ::=\ \ \recorddecl{r}{\mathbf{l_{i}}=e_{i}} \mid \dots 
\hspace{10ex}
e\ \ ::=\ \ \mathbf{r.l} \mid \dots 
\]
Record definitions are stored in a new label component \cat{R} which is a pairs category containing objects mapping record identifiers to record terms (where a record term is $\{\mathbf{l_{i}}=e_{i}\}$, with $i$ ranging here over the list of record elements). Extend the previous labels category with:
$\labtrans{R}{\cat{R}}$.
The transition rules for the two new programming constructs are:
\begin{prooftree}
\AxiomC{$\mathbf{r}\not\in\rho_{r}$}
\UnaryInfC{$\recorddecl{r}{\mathbf{l_{i}}=e_{i}} \transition{\{R=\rho_{r}\restlabel R=\rho_{r}[\mathbf{r}\mapsto \{\mathbf{l_{i}}=e_{i}\}]\}} \nilValue$}
\end{prooftree}
\begin{prooftree}
\AxiomC{$\rho_{r}(\mathbf{r})=\{\mathbf{l_{i}}=e_{i}\},\ \exists i:\mathbf{l_{i}}=\mathbf{l}, e_{i}=e$}
\UnaryInfC{$\mathbf{r.l} \transition{\{R=\rho_{r}\restlabel\}} e$}
\end{prooftree}
%

The rules above give a ``lazy'' semantics for records, where the evaluation of the expressions is postponed until the record label is referenced. This is similar to inlining constructs, as e.g.\ in the Promela \cite[ch.3]{spinBook}. Moreover, these rules implement a small-step semantics. 
Big-step or eager semantics could also be given.

The choice of syntax for the records is biased by our goal to reach
\proteusSrcLang. Nevertheless, using the theory we presented, one may
give semantics to more complex records  
such as those in e.g.\ \cite[Chap.9]{huttel10SOSbook}.

\subsection{Conditional construct}\label{example_if}
The conditional construct, of sort \emph{Statement}, taking as parameters a term of sort expression and two terms of sort statement, can be added to any of the signatures from before; here $\SigmaExec_{\ref{example_assignment}}\subset\SigmaExec_{\ref{example_if}}$.
\[
s\ \ ::=\ \ \ifthenelseStatement{e}{s_{1}}{s_{2}} \mid \dots 
\]
The semantics does not rely on any particular form of the label categories.
\begin{prooftree}
\AxiomC{$e\transition{X}\mathbf{true}$}
\UnaryInfC{$\ifthenelseStatement{e}{s_{1}}{s_{2}} \transition{X} s_{1}$}
\DisplayProof\hspace{3ex}
\AxiomC{$e\transition{X}\mathbf{false}$}
\UnaryInfC{$\ifthenelseStatement{e}{s_{1}}{s_{2}} \transition{X} s_{2}$}
\end{prooftree}

\subsection{Comparison with \proteusSrcLang}

By now we have reached the language \proteusSrcLang\ of \cite{StoyleHBSN07mutatis} (omitting reference constructs, which could be added following \cite[Sec.4.2]{mosses99concurrentML}). We add the upgrade construct in Section~\ref{subsec_DSOS_proteus}.
We have used single variable identifiers above, but this can be easily generalized to lists.
Moreover, since we investigate only semantic aspects in this paper (i.e., no typing systems), we assume only syntactically correct programs, including static typing. Discussions about typing over MSOS and DSOS are relegated to Section~\ref{sec_typing_proteus}.

The transition rules that we gave for \proteusSrcLang\ used a label category formed of three components: \cat{S}, \cat{F}, and \cat{R}. In \cite{StoyleHBSN07mutatis} the semantics of \proteusSrcLang\ keeps all these information in one single structure called \textit{heap}. The separation of this structure that we took does not impact the resulting semantic object obtained for \proteusSrcLang\ in \cite[Fig.12]{StoyleHBSN07mutatis}. 

\begin{proposition}[conformance with \proteusSrcLang\ semantics]\label{prop_conform_Proteus_sem}
Considering reductions $\Rightarrow$ to be either a compilation or an evaluation step from \cite[Fig.12]{StoyleHBSN07mutatis}, and the transitions $\transition{\alpha}$ obtained with the MSOS rules for \proteusSrcLang, we have that
\[
\Omega,H,e \Rightarrow \Omega,H',e' \mbox{\ \ iff\ \ } e\transition{\alpha}e' \mbox{ with }
\]
\[
\source{\alpha}=(\rho_s,\rho_f,\rho_r),\target{\alpha}=(\rho'_s,\rho'_f,\rho'_r),H=\rho_s\cup\rho_f\cup\rho_r,H'=\rho'_s\cup\rho'_f\cup\rho'_r.
\]
When we add types in Section~\ref{sec_typing_proteus} then the typing environment $\Omega$ may change and will be captured by the types label $\cat{TY}$ on the morphisms: $\Omega=\rho_{\mathit{ty}},\Omega'=\rho'_{\mathit{ty}}$.
\end{proposition}

\begin{proof}
The proof of this proposition essentially uses the relation between standard labelled transition systems and the arrow-labelled transition systems of the MSOS \cite[Prop.3\&4]{Mosses04modularSOS}. 
Here we are looking at the particular rules of \proteusSrcLang. It is not difficult to see that the changes (and inspections) to the heap that are made in the original rules of \cite[Fig.12]{StoyleHBSN07mutatis} are matched by the ones mentioned on the arrows of the MSOS rules given above.

We first correlate the functional syntax used by \proteusSrcLang\ with our more imperative definitions from Sections~\ref{example_noLabels}-\ref{example_if}.
The constructs for sequential composition from Sec.~\ref{example_noLabels} are encoded in the functional style of \proteusSrcLang\ using multiple applications of the $\mathbf{let}$ construct.
Our syntax for the $\mathbf{let}$ construct (Sec.~\ref{example_variables}) as well as for variable definition (Sec.~\ref{example_assignment}) and function definition (Sec.~\ref{example_functions}) are the same as in \proteusSrcLang, albeit looking more imperative than functional (e.g.: instead of the \proteusSrcLang\ notation $\mathbf{z}\mapsto\lambda(x).e$ for function definition we use $\funcdecl{f}{\mathbf{x}}{s}$ with $\mathbf{f}$ as the $\mathbf{z}$ and $s$ as the $e$).
For records we chose in Sec.~\ref{example_records} to name them $\mathbf{record\ r}$ and to use this name in projections $\mathbf{r.l}$, whereas \proteusSrcLang\ uses just expressions when doing projections, which for us is the body of the record $\{\mathbf{l_{i}}=e_{i}\}$.
The $\mathbf{if}$ statement from \proteusSrcLang\ uses as test the comparison of two expressions, whereas in Sec.~\ref{example_if} we use only one expression and let the rules decide that the $\mathbf{if}$ is executed only when this expression evaluates to a Boolean.

We also correlate the transition rules of \proteusSrcLang\ from \cite[Fig.12]{StoyleHBSN07mutatis} with our rules from Sections~\ref{example_noLabels}-\ref{example_if}.
\proteusSrcLang\ uses evaluation contexts \cite[Fig.11]{StoyleHBSN07mutatis} and one rule \textsc{(cong)} for context reductions in \cite[Fig.12]{StoyleHBSN07mutatis}.
We achieve the same effect by adding for each construct explicit rules that evaluate expressions until their final value form. This is not new, e.g., \cite{mosses99concurrentML} does this in the MSOS style for a functional language and explains well in \cite[Sec.5.1]{mosses99concurrentML} the correlations with other related styles of semantics including evaluation context reduction.
For the $\mathbf{if}$ construct \proteusSrcLang\ uses one evaluation context $\mathbf{let\ } z=E \mathbf{\ in\ }e$ which is meant to ensure that the bound variable $z$ the expression is evaluated to a value, after which the corresponding rule \textsc{(let)} from \cite[Fig.12]{StoyleHBSN07mutatis} is applicable. In our case, the left rule from Sec.~\ref{example_variables} corresponds to the evaluation context, whereas the right rule is the same as in \proteusSrcLang.
For our variable declarations and assignments in Sec.~\ref{example_assignment} the left rule corresponds to first evaluating the expression to a final value, whereas the right rule corresponds to the last compilation rule of \cite[Fig.12]{StoyleHBSN07mutatis} where the heap is updated (in our case the label $S$ is involved).
For functions our right-most rule from Sec.~\ref{example_functions} corresponds exactly to the compilation rule for functions from \cite[Fig.12]{StoyleHBSN07mutatis} (the remaining compilation rule from \cite[Fig.12]{StoyleHBSN07mutatis} is not applicable to us because we do not have typing information).
Our other two rules correspond, the left-most one to the evaluation contexts for function applications from \cite[Fig.11]{StoyleHBSN07mutatis}, whereas the middle one to the rule \textsc{(call)} from \cite[Fig.12]{StoyleHBSN07mutatis}.
In Sec.~\ref{example_records} we chose to give a lazy semantics to records, where we store and return the expressions corresponding to some record entry, whereas \proteusSrcLang\ gives an eager semantics where they store and return the corresponding values. For this \proteusSrcLang\ keeps in the heap records with values, whereas we keep in the label component \cat{R} records as defined with their original expressions. Moreover, \proteusSrcLang\ uses evaluation contexts for records to produce their corresponding values, whereas we do not. However, it is straightforward to give eager rules in the MSOS style; we only need to add similar as before rules for evaluating expressions until their final values (corresponding to the contexts of \proteusSrcLang) and then rules similar to the current ones but which work on values instead of expressions.
Our two rules from Sec.~\ref{example_if} first evaluate the test expression, and if it evaluates to a Boolean, one or the other of the branches is taken as the continuing statement. This matches the two transition rules \textsc{(if-t)} and \textsc{(if-f)} of \proteusSrcLang\ from \cite[Fig.12]{StoyleHBSN07mutatis} which work only on values, and also explicit the evaluation contexts from \cite[Fig.11]{StoyleHBSN07mutatis} together with the evaluation context reduction rule \textsc{(if-t)} from \cite[Fig.12]{StoyleHBSN07mutatis} for this statement. We provided big-step style rules only for exemplification purposes, whereas small-step style rules would be similar to what we did for the other previous constructs.
The other rules from \cite[Fig.12]{StoyleHBSN07mutatis} are not applicable, especially rule 2 is for coercions, which we do not consider, rules 5-6 are for references and are similar to \cite[Sec.4.2]{mosses99concurrentML}, whereas rules 10 and 12 are for updates, which we consider further down.
\end{proof}

\section{Dynamic SOS}\label{sec_DSOS}

To give intuitions for Dynamic SOS consider the program term as acting
on a data structure during its execution, like a store or a heap,
or a configuration reflecting a distributed run-time
  environment. Classical operational semantics describes how each
programming construct changes these data structures (or uses the
information stored in them). The dynamic upgrades use upgrade data that is seen as
coming from outside the program, being controlled by an external entity. 
It is irrelevant for the upgrade programming construct where or how the upgrade data appears. What is important though is how the upgrade construct uses the upgrade data (e.g., to change the program's state) and when during the execution of the program. This is described through the semantics of the upgrade constructs and is ensured safe through static analysis (like for any other programming constructs).
Implementing a way to insert upgrade data can be done in various ways, independent of the semantics of the upgrade constructs; e.g., in \creol\ a pool of messages in maintained for communications between the programming objects (i.e., part of the way a program executes), and this is also used for class upgrades by inserting into the pool a special upgrade message which is not used by any programming constructs, but only by the upgrade mechanism.

DSOS considers that there is a separate data structure containing
information about upgrades. This upgrade data structure is changed by
the external entity at any point in the execution of the program, and
the program may only inspect it. The program can decide at which
points in the execution it is safe to do an upgrade. The upgrade
operation takes information from the upgrade data and changes 
the data structures that the program maintains. Therefore, this may
change the behaviour of the program.
This is similar to the \proteusSrcLang\ update mechanism.

Since the upgrade points are decided by the program, upgrade programming constructs can be added to the language. A programmer can use these, or a tool can detect program points, and insert such upgrade constructs in the code as necessary. The semantics of an upgrade construct essentially describes how the upgrade changes the data structures of the program.

These ideas are simple and capture only how the semantics of upgrades should be thought and defined. Complications may appear in the definition of the actual update functions of the data structures, as well as in the analysis technique of the programming language for detecting the upgrade points. These also interact with the typing system. Much of the related works on dynamic upgrading constructs \cite{MalabarbaPGBB00,DrossopoulouDDG02,BoyapatiLSMR03,JohnsenOS05dynamicClasses,AjmaniLS06,StoyleHBSN07mutatis,BiermanPN08} focus on these aspects, which are usually developed on top of the semantics. We discuss typing aspects in Section~\ref{sec_typing_proteus}.

Dynamic SOS builds on the modular approach from the previous sections by incorporating the following aspects. 
\begin{enumerate}
\item The arrow-labelled transition system is enriched by adding new kinds of transitions labelled not with morphisms, but with endo\-functors. 

\item In consequence, the syntax for writing transition rules is enriched to use endo\-functors. 

\item The label transformer is enriched accordingly, and also the label categories that we use. Defining the endo\-functors, though, is something we are already familiar with, as we shortly see.

\item The program syntax is assumed to have programming constructs denoting upgrade points of various kinds, the semantics of which are given with the endo\-functors. 
\end{enumerate}

\begin{definition}[upgrade transition systems]\label{def_updTransSys}
An \emph{upgrade transition system} (UTS) is $(\Gamma,L,\transition{})$ with $\Gamma$ the set of program terms and $L=\morphisms{\cat{A}}\cup\morphisms{\mathit{End}(\barecat{A})}$ the set of labels with $\mathit{End}(\barecat{A})$ the category of endofunctors over $\barecat{A}$ from Definition~\ref{def_functor}, and $\barecat{A}$ having the same objects as $\cat{A}$, i.e., $\objects{\barecat{A}}=\objects{\cat{A}}$.
We call the transitions labelled by endo\-functors, \emph{jumps}, and distinguish them by labeling with capital letters $E\in\morphisms{\mathit{End}(\cat{A})}$. The other transitions are called \emph{steps}.
A \emph{computation} in UTS is a possibly infinite sequence of transitions from $\Gamma\times L\times\Gamma$, i.e., $t_{0}\transition{l_{0}}t_{1}\transition{l_{1}}t_{2}\dots$, starting with a step, and restricted in the following sense: for $\pi$ denoting the sequence of labels in a computation, i.e., defined with the grammar
$\pi ::= \ \transition{\alpha}\ \mid\ \pi \transition{\alpha}\ \mid\ \pi \transition{E}$,
the sequencing of two transitions is allowed only when their labels respect the following: 
\[
\pi \transition{\alpha} \mbox{\ \  iff\ \ \ } \target{\pi}=\source{\alpha}
\]
with the $\target{(\cdot)}$ defined for computations $\pi$ inductively as 
\[
\target{(\transition{\alpha})}=\target{\alpha} \hspace{3ex} \target{(\pi\transition{\alpha})}=\target{(\transition{\alpha})} \hspace{3ex} \target{(\pi\transition{E})}= E(\target{\pi}).
\]
%
\end{definition}

\begin{corollary}
When there are no jumps, a computation in UTS is defined exactly as for ALTSes.
\end{corollary}

Requiring a computation to start with a step transition captures our intuition that dynamic upgrades may happen only during the execution of the program, but not before it starts.
Note that the endofunctors are used only wrt.\ their applications on the objects of the category, disregarding their application on morphisms. This is why we only specify that $\objects{\barecat{A}}=\objects{\cat{A}}$. Thus, we can have as $\barecat{A}$ any version of \cat{A} with more or less morphisms.

\begin{definition}[upgrade label transformers]\label{def_labtransUpd}
Consider a second indexing set \indexU\ disjoint from \indexE. The \emph{upgrade label transformer} is defined the same as the label transformer from Definition~\ref{def_labtrans}, but using the upgrade indexes $j\in\indexU$. The $\ulabtrans{j}{\cat{U}}$ maps a category \cat{A} to a product category $\cat{A}\prodcat\cat{U}$, where \cat{U} may only be a discrete category.
\end{definition}

The \cat{U} categories are called the \emph{upgrade components} of the labels. These are discrete because the program is not supposed to change the upgrade information, i.e., any morphisms on the transitions would include only identity morphisms for the upgrade components. Because of the disjointness of the indexing sets, the same get operation from before is still applicable, and existing transition rules are not affected by the addition of an upgrade component. 
In essence, the upgrade label transformer is a special case of the label transformer, i.e., uses a disjoint set of indexes \indexU\ and only discrete categories \cat{U}.

Because the upgrade components are discrete categories, when referring to an upgrade component of a morphism label we in fact refer to the current upgrade object. 
Modularity is not disturbed, and new data categories may be added with the label transformer in the same way, without any interference with the upgrade components. 

The semantics of dynamic software upgrades is given in terms of endo\-functors on the product category.
These endo\-functors are obtained from combining \textit{basic endo\-functors}, which are defined in terms of only some of the data and the upgrade components. To understand how the endo\-functors are obtained and how the basic ones should be defined, we first give some properties specific to the kinds of categories that we use.
%
%

\begin{proposition}\label{prop_propertiesDiscretPair}
Properties for label categories and their products.

\begin{enumerate}
\item In discrete or pairs categories morphisms are uniquely defined by the objects.

\item\label{propertiesDiscretPair_2} Let \cat{A} and \cat{B} be both either pairs or discrete categories. In the category returned by the label transformer $\labtrans{i}{\cat{B}}(\cat{A})$ the morphisms are uniquely defined by the objects.

\item\label{propertiesDiscretPair_3} Let \cat{A} and \cat{B} be both either pairs or discrete categories and \cat{C} a monoid category. In the category returned by the label transformer $\labtrans{j}{\cat{C}}(\cat{A})$, as well as in $\labtrans{i}{\cat{B}}(\labtrans{j}{\cat{C}}(\cat{A}))$, each morphism is uniquely determined by the objects up to the morphism components coming from the monoid category; i.e., when the monoid components are projected away.
\end{enumerate}
\end{proposition}

\begin{proof}
Verifying the three properties is an easy exercise in category theory.
\end{proof}

For discrete or pairs categories the endo\-functors have a special property, they are completely defined by their application to the objects of the category only.

\begin{proposition}\label{prop_endo_prodPairs_prop}
Let \cat{A} be a discrete or a pairs category, and $F:\cat{A}\rightarrow\cat{A}$ an endo\-functor on \cat{A}. $F$ is completely defined by its application to the objects of \cat{A}.
\end{proposition}

\begin{proof}
Consider that for $F$ we know how it is applied to the objects in $\objects{\cat{A}}$. Consider one morphism $o\transition{\alpha}o'$, which is uniquely defined by the two objects $o,o'$ (which may also be the same object, in a discrete category). The functor associated to this morphism is the following morphism from \cat{A}: $F(\alpha)=(F(o),F(o'))$ which is the unique morphism from $F(o)$ to $F(o')$, hence respecting the requirements from Definition~\ref{def_functor} of being a functor.
\end{proof}

However, Proposition~\ref{prop_endo_prodPairs_prop} talks about products of only pairs categories or products of only discrete categories. 
Whereas, the product of a discrete with a pairs category is different since there may be tuples of objects with no morphism between them. This is an issue when putting together an upgrade component, which is always discrete, and a pairs data component. 
To be in line with our intuition that an upgrade operation should be arbitrarily definable and dependent on both the upgrade and the data objects, we will define endo\-functors on \emph{discretized categories}.


\begin{definition}\label{def_discretizedCat}
A \emph{discretized category} \discrete{\cat{A}} is  obtained from a category \cat{A} by removing all non-identity morphisms. 
\end{definition}

Endofunctors are meant to describe how upgrade information from the objects of the \cat{U} components change the objects from the data components, thus defining a correspondence between the data before and after some upgrade, for any upgrade information. The result above suggests that for pairs or discrete categories, defining such endo\-functors resorts to only defining their application on the objects of the category (i.e., a total function). These objects are normally tuples involving both upgrade and data objects.

In general endofunctors must also relate the morphisms, which restricts their definition. When a pairs category is coupled with a discrete upgrade category then the endofunctor definition on the objects must be made in such a way that morphisms on the pair category are somewhat preserved. This would, for example, not allow to freely change the upgrade object, e.g., since $(p_{1},u)\transition{(\alpha,id_{u})}(p_{2},u)$ we cannot define $F(p_{1},u)=(p_{3},u_{1})$ and $F(p_{2},u)=(p_{4},u_{2})$ because there is not morphism between these last two. This example is intuitive when doing incremental upgrades using only part of the upgrade information that disappears after the upgrade operation.
It is interesting to study what kinds of practically useful upgrades can be defined if endofunctors can be defined on whatever label categories.

%

We will work in this paper with simple endo\-functors, as in Proposition~\ref{prop_endo_prodPairs_prop}, for which the application on the objects is enough. However, DSOS should handle more complex categories, where the action of the endo\-functors on the morphisms may also be relevant, e.g., monoid categories which are used for handling errors in \cite[Sec.3.7]{Mosses04modularSOS}.
A monoid category intuitively defines ``labels'' on transitions which are used (often in process algebras) to define which transitions (with what labels) are allowed from which program terms. However, the program term itself does not use this information.
In the presence of monoid labels we can still define the endofunctors only on the objects, and have a natural definition of the corresponding morphisms, i.e., matching the monoid part of the morphism pair (according to Proposition~\ref{prop_propertiesDiscretPair}(\ref{propertiesDiscretPair_3})).
%
%
Since we work with simple categories in this paper, one question (which we discuss more in the further work Section~\ref{subsec_further_work}) is whether such endo\-functors can be encoded into simple morphisms of potentially different categories.

\begin{notation}
For some indexing set $I\subset\indexE$ (or $I\subset\indexU$)
we denote
by $\cat{D}_{I}$ (respectively $\cat{U}_{I}$) the product category $\prodcat_{i\in I}\cat{D}_{i}$ obtained using the (upgrade) label transformer using the indexes from $I$ attached to the respective category component.
\end{notation}

\begin{definition}[basic endofunctors]\label{def_basicUpgrEndo}
For a product category $\cat{D}_{I}\!\prodcat\!\cat{U}_{K}$ obtained using \labtrans{}{}\ and \ulabtrans{}{}, consider the discrete version of this to be $\discrete{\cat{D}}_{I}\!\prodcat\!\cat{U}_{K}$, and define a \emph{basic upgrade endo\-functor} $\basic{E}$ as a total function over the objects of this category.
\end{definition}

It remains to see how to combine basic endo\-functors from acting locally, on label components, to one single endo\-functor on the whole label category. We essentially make pairs of endo\-functors over the product of categories.

\begin{proposition}[endo\-functors as morphisms]\label{prop_endo_as_morphisms}
Consider two categories \cat{A} and \cat{B} with $\catendo{\cat{A}}$ and $\catendo{\cat{B}}$ as in Definition~\ref{def_functor}.
Define the product of two such categories $\catendo{\cat{A}}\prodcat\catendo{\cat{B}}$ to have one object $(\cat{A},\cat{B})$ and morphisms the pairs of morphisms from the two categories.

\begin{enumerate}
\item Any morphism $(E_{\cat{A}},E_{\cat{B}})$ in the product $\catendo{\cat{A}}\prodcat\catendo{\cat{B}}$ is an endo\-functor on $\cat{A}\prodcat\cat{B}$ which takes any object $(a,b)\in\objects{\cat{A}\prodcat\cat{B}}$ to an object $(E_{\cat{A}}(a),E_{\cat{B}}(b))$ and any morphism $(\alpha,\beta)$ to $(E_{\cat{A}}(\alpha),E_{\cat{B}}(\beta))$. 

\item If the categories \cat{A} and \cat{B} have the property of Proposition~\ref{prop_propertiesDiscretPair}(\ref{propertiesDiscretPair_2}), like discrete or pairs categories, and their products, then the pairs of endo\-functors are also completely defined by their application on the objects.
\end{enumerate}
\end{proposition}

\begin{proof}
The proof uses basic notions of category theory, and becomes even easier in the light of the proof of Proposition~\ref{prop_endo_prodPairs_prop}.
\end{proof}

Thus, the paired endo\-functors have the same properties as the component endo\-functors, and their behaviour is defined by their component endo\-functors.

The only requirement that we ask of the endo\-functors is that once an \emph{information-less object} is reached, then no more change of data objects can be performed. This is a termination condition where inaction from the functor is required. Intuitively, an upgrade should not change the data of the program if there is no upgrade information.

\begin{definition}
For any upgrade category $\cat{U}$ we identify at least one (or more) objects as being \emph{information-less object}, and denote such objects with a ``bottom'' symbol at subscript, e.g., $o_{\bot},u_{\bot}$.
\end{definition}

The categories that we encountered in our examples all have information-less objects, e.g.: 
\begin{itemize}
\item when the underlying objects are \textit{sets} then $o_{\bot}$ is the $\emptyset$; 
\item when the underlying objects are \textit{partial functions} then $o_{\bot}$ is the minimal partial function completely undefined; 
\item for a category with a single object, like the \textit{monoid category}, then this is considered to be the $o_{\bot}$; 
\item for a \textit{product of categories} then the pairing of all the corresponding $o_{\bot}$ is the information-less object. 
\end{itemize}
All examples above have the set of objects \textit{equipped with a partial order}, in which case the information-less objects are the minimal objects in the partial order.

\begin{definition}[no sudden jumps]\label{def_functorRestrict}
An endo\-functor $E$ on $\cat{D}\prodcat\cat{U}$ is said to \emph{have no sudden jumps} iff \ 
$\forall u_{\bot}\in\objects{\cat{U}}: E((d,u_{\bot}))=(d,u_{\bot})$.
Both \cat{D} and \cat{U} can be arbitrary product categories.
\end{definition}

All endo\-functors that we give as examples in this paper can be easily checked to have no sudden jumps, i.e., are inactive on information-less objects.

\begin{definition}[extending endo\-functors]\label{def_upgrade_endo}
For a product category $\cat{D}_{I}\!\prodcat\!\cat{U}_{K}$ obtained using \labtrans{}{}\ and \ulabtrans{}{}, define a basic upgrade endo\-functor $\basic{E}$ as in Definition~\ref{def_basicUpgrEndo} over some part of this category, i.e., over $\discrete{\cat{D}}_{I'}\!\prodcat\!\cat{U}_{K'}$, with $\emptyset\!\neq\!I'\!\subseteq\!I$ and $\emptyset\!\neq\!K'\!\subseteq\!K$. This basic endo\-functor must have no sudden jumps. \emph{Extend $\basic{E}$} to the whole product category by pairing it with the identity endo\-functor on the remaining component categories, as in Proposition~\ref{prop_endo_as_morphisms}.
%
\end{definition}

Note that extending with identity endofunctors can be done over arbitrary kinds of label categories, i.e., the restriction to discretized category is needed only for defining the basic endofunctors.
Note that any basic endo\-functor is defined over a \textit{variant} of $\barecat{D}_{I}\!\prodcat\!\barecat{U}_{K}$ thus respecting the requirements from Definition~\ref{def_updTransSys} of UTS. Moreover, any extension is also over a \textit{variant} of the larger product category, though not necessarily over a discrete variant as the basic endofunctors are.

\begin{proposition}[composing upgrade endo\-functors]\label{prop_composing_funct}
For two basic endo\-functors defined on disjoint sets of indexes, their extensions can be composed in any order, resulting in the same endo\-functor on the union of the indexing sets.
\end{proposition}

\begin{proof}
Consider a product category $\cat{D}_{I}\prodcat\cat{U}_{J}$ built with the label transformer over the index sets $I\cup J$. Without loss of generality we we explain the proof for the simpler category $\cat{D}\prodcat\cat{D'}\prodcat\cat{U}\prodcat\cat{U'}\prodcat\cat{K}$ (full proof can be found in the technical report \cite{techRep12Dec}). 
Consider two endo\-functors $E,E'$ built over $\cat{D}\prodcat\cat{U}$ respectively $\cat{D'}\prodcat\cat{U'}$; the disjointness is important. The category \cat{K}\ can be any upgrade or data categories.

Extend each endo\-functor from above to the whole category as in
Definition~\ref{def_upgrade_endo} by pairing it with the identity
endo\-functor on the remaining category; e.g., for $E$ denote
its extension as $\tilde{E}$ to be the product $E\prodcat
\mathit{ID}_{\cat{D'}\prodcat\cat{U'}\prodcat\cat{K}}$. The similar extension for $E'$ is
$\tilde{E}'=E'\prodcat \mathit{ID}_{\cat{D}\prodcat\cat{U}\prodcat\cat{K}}$. 
Since the identity endo\-functors can be seen as products of smaller identity endo\-functors, we can rewrite the above endo\-functors to: 
$\tilde{E}=E\prodcat\mathit{ID}_{\cat{D'}\prodcat\cat{U'}}\prodcat\mathit{ID}_{\cat{K}}$ and $\tilde{E}'=E'\prodcat \mathit{ID}_{\cat{D}\prodcat\cat{U}}\prodcat \mathit{ID}_{\cat{K}}$. 
We have been relaxed with the notation for the products, but care must be taken for the order of the arguments, so one would write $\tilde{E}'$ as $\mathit{ID}_{\cat{D}\prodcat\cat{U}}\prodcat E'\prodcat \mathit{ID}_{\cat{K}}$.

We need to show that 
\[
 \tilde{E}'\composition\tilde{E} = \tilde{E}\composition\tilde{E}' = E\prodcat E'\prodcat \mathit{ID}_{\cat{K}} .
\]
Pick now two objects from the big category: 
$(d_{1},u_{1},d'_{1},u'_{1},d_{1}^{k})\mbox{ and }(d_{2},u_{2},d'_{2},u'_{2},d_{2}^{k})$.
The morphism between the tuple objects is also a tuple of respective morphisms $(\alpha_{d},\alpha_{u},\alpha'_{d},\alpha'_{u},\beta)$. Apply now the endo\-functor $\tilde{E}$ to obtain tuples of objects $(E(d_{1},u_{1}),d'_{1},u'_{1},d_{1}^{k})$ and $(E(d_{2},u_{2}),d'_{2},u'_{2},d_{2}^{k})$, and morphism $(E(\alpha_{d},\alpha_{u}),\alpha'_{d},\alpha'_{u},\beta)$. 
To this apply the second endo\-functor to obtain objects $(E(d_{1},u_{1}),E'(d'_{1},u'_{1}),d_{1}^{k})$ and $(E(d_{2},u_{2}),E'(d'_{2},u'_{2}),d_{2}^{k})$, and morphism $(E(\alpha_{d},\alpha_{u}),E'(\alpha'_{d},\alpha'_{u}),\beta)$.

It is easy to see that for the other composition $\tilde{E}\composition\tilde{E}'$ we would obtain the same objects and morphism. Moreover, these are independent of the monoid categories that are subject only to the identity endo\-functor $\mathit{ID}_{\cat{K}}$.

From the above it is easy to see how one could first make the product of the two endo\-functors $E\prodcat E'$ and afterwards extend this to the whole category, as $E\prodcat E' \prodcat \mathit{ID}_{\cat{K}}$, and the result of the application of this product results in the same objects and morphisms as the compositions above.
\end{proof}

Proposition~\ref{prop_composing_funct} ensures modularity of 
 Dynamic SOS as follows. 
One defines a basic endo\-functor for some dynamic upgrade construct, and this is never changed upon addition of other dynamic upgrade constructs and their upgrade categories and related endo\-functors. Moreover, the method of \textit{extending} the basic endo\-functors with the identity functor on the rest of the indexes, from Def.~\ref{def_upgrade_endo}, ensures modularity when new data or upgrade components are added by the label transformers.

When designing a programming language the label transformers may be applied on an already used index, resulting in changing the respective category component, e.g.:
\begin{itemize}
\item we may change a read-only component into a read/write component.
\item we may decide to have more upgrade functors on one particular component, i.e., to define a new way of updating, maybe needed by a new programming constructs.
\item we may leave one functor unspecified, as the identity functor, and at a later point add a proper functor for the specific component.
\end{itemize}


The encapsulation construction from Section~\ref{sec_encapsulation} can be applied to endo\-functors as well. This is expected, because if we encapsulate the categories on which the endo\-functors act, then the endo\-functors would become undefined. While by encapsulating them the endo\-functors would be preserved. Once encapsulated, we may refer to the endo\-functors using the object identifiers, the same as we were referring to the localized data components.

Each endo\-functor is matched (using a transition rule) by a dynamic upgrade construct in the programming language, for which it captures the desired upgrade mechanism; this is exemplified in the next section.

Much of the work in \cite{StoyleHBSN07mutatis} is concerned with analyzing \proteusSrcLang\  program terms to automatically insert upgrade constructs at the appropriate points in the program where the upgrade would not cause type errors. The same analyses can be done also when the language is given a DSOS semantics. 

A similar, but rather coarse analysis of upgrade points is done for the concurrent object-oriented language \creol\ of \cite{JohnsenOS05dynamicClasses}, where acceptable upgrade points are taken to be those execution points of an object where it is ``idle'' (called \textit{quiescent} states in \cite{JohnsenKY09}, where the processor has been released and no pending process has been activated yet). A more fine-grained analysis in the style of \cite{StoyleHBSN07mutatis} could be carried out, but it would be necessarily more complex because of the concurrency and object-oriented aspects, and also because of the special asynchronous method calls and late bindings. 
Such an analyses for the \creol\ language is beyond the scope of this paper.

\subsection{Exemplifying DSOS for \proteusSrcLang}\label{subsec_DSOS_proteus}

For this section knowledge of \proteusSrcLang\ \cite{StoyleHBSN07mutatis} is not needed since our discussions will use only standard programming languages terminology. Nevertheless, we constantly refer to \proteusSrcLang\ and the work in \cite{StoyleHBSN07mutatis} for completeness and guidance for the familiar reader. 

The transition rules that we gave for \proteusSrcLang\ constructs \cite[Fig.2]{StoyleHBSN07mutatis} in Section~\ref{sec_MSOS_proteus} used a label category formed of three components: 
\cat{S} with objects mapping variable identifiers to values, 
\cat{F} with objects mapping function names to definitions of functions as lambda abstractions, 
\cat{R} with objects mapping record identifiers to definitions of records. 
In \cite[Sec.4.3]{StoyleHBSN07mutatis} the semantics of \proteusSrcLang\ keeps all these information in one single structure called \textit{heap}. The separation of this structure that we took does not impact the resulting semantic object, as one can check against \cite[Fig.11]{StoyleHBSN07mutatis}. 
Our choice was made with the intention to obtain a more clear separation of concerns, where we can see from the transition rules which programming construct works with what part of the program state, and in what way it interacts with the other parts. One can easily correlate our rules with the ones in \cite[Fig.12]{StoyleHBSN07mutatis}.

Four kinds of update information are present in \proteusSrcLang. In this exemplification we treat only the two not related to types, i.e., the update and the addition of new bindings to the heap. Updating or adding new types is discussed in  Section~\ref{sec_typing_proteus}. 
In \cite[Fig.11]{StoyleHBSN07mutatis} the update information comes in the form of a partial mapping from top-level identifiers to values (we omit the types for now). This update information follows the same structure as the heap. At any time point, in the heap we can see the identifiers separated into variables, function names, or record names; the values being either basic values for variables, lambda abstractions containing the function body, and record definitions. It is easy to see that we get the corresponding structures as the objects in our categories \cat{S}, \cat{F}, respectively \cat{R}. Therefore, the corresponding update categories are: $\cat{U}_{\cat{S}}$, $\cat{U}_{\cat{F}}$, and $\cat{U}_{\cat{R}}$, discrete categories containing the same objects as respectively \cat{S}, \cat{F}, and \cat{R}.


\proteusSrcLang\ uses a single update construct, which marks points in the program where updates can take place. 
We separate these update constructs into three kinds, each dealing with variables, functions, or records. Thus, our update signature $\SigmaUpdate$ contains: 
\[
s\ \ ::=\ \ \upgradeTermSetKind{\Delta}{v} \mid \upgradeTermSetKind{\Delta}{f} \mid \upgradeTermSetKind{\Delta}{r} \mid \dots 
\]
where $\Delta$ is a set of identifiers of respectively variables, functions, or records. 

Having defined the update categories, it remains to define the corresponding endo\-functors. Since the endo\-functors for our special categories can be given solely by their application on the set of objects, we define one endo\-functor for each update category as a function applied to pairs of data and update objects, e.g., from $\objects{\cat{S}}\times\objects{\cat{U}_{\cat{S}}}$. 
Define an update transition rule as:
%
\begin{prooftree}
\AxiomC{\phantom{$\ \transition{U} \ $}}
\UnaryInfC{$\upgradeTermSetKind{\Delta}{v} \transition{E_{\Delta}^{v}} \nilValue$}
\end{prooftree}
with $E_{\Delta}^{v}\in\morphisms{\catendo{\cat{S}\prodcat\cat{U}_{\cat{S}}}}$ an endo\-functor on the product category $\cat{S}\prodcat\cat{U}_{\cat{S}}$, defined below the same as in \cite[Fig.13]{StoyleHBSN07mutatis} but restricted to consider only those variable identifiers specified in $\Delta$ and remove them from the update objects. Thus, both the data object and the update object may be changed by an endo\-functor.
For one store object $\rho$ of $\objects{\cat{S}}$ and one update object $\rho_{u}$ of $\objects{\cat{U}_{\cat{S}}}$ the endo\-functor $E_{\Delta}^{v}$ changes $\rho_{u}$ by removing all the mappings for the variable identifiers appearing in $\Delta$; and changes $\rho$ by replacing all mappings from variable identifiers appearing in $\Delta$ with the corresponding ones from $\rho_{u}$:
$$E_{\Delta}^{v}(\rho,\rho_{u})=
\left\{\begin{array}{lr} 
(\rho[\mathbf{x}\mapsto\rho_{u}(\mathbf{x})\mid \mathbf{x}\in\Delta\cap\rho_{u}],\,\rho_{u}\!\!\setminus\!\Delta) & \mbox{\ \ if }\mathit{dom}(\rho_{u})\cap\Delta\neq\emptyset\\
(\rho,\rho_{u}) & \mbox{ otherwise} 
\end{array} \right.
$$

For the typed case we would need a more complex safety check which can be taken from \cite[Fig.24]{StoyleHBSN07mutatis} where it is called $\mathsf{updateOK}(-)$ and which also checks that the update information is well typed, not only that all needed identifiers are part of the update, as we did here. 
In fact one could do any kind of sanity checks of the update information against the data. However, at the level of the functor definition one does not have access to the program term, as is done in \cite[Fig.16]{StoyleHBSN07mutatis}. Any such information must either be put in the data part (e.g., as done when having threads), or be dealt with statically, as is done in \cite[Sec.5]{StoyleHBSN07mutatis} to obtain the definition of $\mathsf{updateOK}(-)$. 

The definition of the endo\-functors is outside the category theory framework of Dynamic SOS because these depend solely on the objects of the data and update categories and their underlying algebraic structure. In consequence, defining endo\-functors requires standard methods of defining functions. This is also the reason why it was immediate to take the definition from \cite[Fig.13]{StoyleHBSN07mutatis} into our setting. 
The contribution of DSOS is not at this level, but it consists of the general methodological framework that DSOS provides, which gives a unified approach to defining dynamic software updates in tight correlation with the normal programming constructs.

The above definition was simple and natural, but more complicated definitions can be devised, especially when the update objects do not have the same structure as the data objects, as is the case for \creol\ in Section~\ref{subsec_ex_DSOS_creol}. 

Our goal in this section was to exemplify the use of DSOS to give semantics to the \proteusSrcLang\ updates without departing from the semantics given in \cite{StoyleHBSN07mutatis}. 
We make this claim more precise in Proposition~\ref{prop_updateOK_asProteus} using notation from \cite{StoyleHBSN07mutatis} but with only a sketch of a proof, since a full proof would require too much background from \cite{StoyleHBSN07mutatis}.

\begin{proposition}\label{prop_updateOK_asProteus}
For any update information $\rho_{u}$, which in \proteusSrcLang\ \cite{StoyleHBSN07mutatis} is denoted $\mathit{upd}$, that updates only variable identifiers, we have that
\[
\Omega,H,\mathbf{update}^{\Delta}\ \transition{\mathit{upd}} \Omega,H',0 \mbox{\ \ iff\ \ } \upgradeTermSetKind{\bar{\Delta}}{v}\transition{E_{\bar{\Delta}}^{v}}\nilValue \mbox{ with }
\]
$\bar{\Delta}$ containing all those identifiers not in $\Delta$, $H=\rho_s\cup\rho_f\cup\rho_r,H'=\rho'_{s}\cup\rho_f\cup\rho_r$, where $E_{\bar{\Delta}}^{v}(\rho_s,\rho_{u})=(\rho'_{s},\rho'_{u})$.
\end{proposition}

This proposition can also be given for full updates of \proteusSrcLang.

\begin{proof}[sketch]
The transition $\transition{\mathit{upd}}$ is defined in \cite[Fig.12]{StoyleHBSN07mutatis} conditioned on the $\mathsf{updateOK}(-)$ safety check. In our case this condition would part of the definition of the endo\-functor $E_{\bar{\Delta}}^{v}$; above our untyped example reduced this check to only a membership check.
The statement of the proposition is only about variable bindings being changed in the heap $H$, which is reflected on the right side in the use of the $\upgradeTermSetKind{\bar{\Delta}}{v}$ construct. To achieve the general updates of \proteusSrcLang\ we can put several of our constructs in sequence to update other entities too. Our choice to have incremental updates can be changed to match the choice in \proteusSrcLang\ exactly; in which case the $\rho'_{u}=\emptyset$.
\end{proof}

\section{Typing aspects over DSOS for \proteusSrcLang}\label{sec_typing_proteus}

This section is meant to substantiate our claims that the typing analyses that make the main results of \cite{StoyleHBSN07mutatis} can also be carried over to the DSOS semantics of \proteusSrcLang. Therefore, this section contains details pertaining to typing from \proteusSrcLang\ which for space reasons could not be included. However, we try to make the general arguments that should be understandable without these details, and an interested reader can then use when closely comparing with \cite{StoyleHBSN07mutatis}.

We need to add \emph{type identifiers} $\mathbf{t}\in\idType$ and \emph{type definitions} $\typedecl{t}{\tau}$, with $\tau$ being basic types, record, functions, or reference types, as in \cite[Fig.2]{StoyleHBSN07mutatis}.
We work with a new label category \cat{TY}, which has type environments $\objects{\cat{TY}}=\idType\rightharpoonup\tau$ as objects, mapping type names to type definitions. This pairs category is attached to the existing labels using $\labtrans{Ty}{\cat{TY}}$. A transition rule would update the type environment consuming a type definition, similar to what we did with variable definitions in Section~\ref{example_assignment}.
Up to now we followed the modularity principle and none of the previous rules need to be changed.
However, when we add type information in the syntax for variable and function definitions we need to change the respective rules too; this is inevitable as the program terms change.
For the label categories there are two options: one more economical, chosen in \proteusSrcLang, where the object of the label categories would map identifiers to tuples of type and value; and a second more modular option, to add new label categories mapping the respective identifiers to their types alone. These categories are treated by the respective changed rules; e.g., the label $\labtrans{Ft}{\cat{FT}}$, which has objects $\objects{\cat{FT}}=\idFunc\rightharpoonup\tau$, is used in the changed rule from Subsection~\ref{example_functions}:
%
\begin{prooftree}
\AxiomC{\phantom{$\mathbf{f}\not\in\rho_{f}$}}
\UnaryInfC{$\funcdecl{f}{\mathbf{x}:\tau_{1}}{s:\tau_{2}} \transition{\{F=\rho_{f},Ft=\rho_{t}\restlabel F=\rho_{f}[\mathbf{f}\mapsto \lambda(\mathbf{X}).s],Ft=\rho_{\mathit{ft}}[\mathbf{f}\mapsto(\tau_{1}\rightarrow\tau_{2})]\}} \nilValue$}
\end{prooftree}
The compilation procedure from \cite[Sec.4.2]{StoyleHBSN07mutatis}, which inserts type coercions, is analogously done over DSOS as it makes no use of the semantics definition, but only of the programming language syntax and typing. In this way the program code can be annotated with $\mathbf{con_{t}}$ and $\mathbf{abs_{t}}$ at those points where the type name $\mathbf{t}$ is known to be further used concretely respectively abstractly. The update operation from \cite[Fig.13]{StoyleHBSN07mutatis} changes (besides the data) also the remaining program code, using type transformers, to make any abstract use of a type into the correct new type. We can avoid this update of the remaining program code by adding two new rules and one label component to deal with statements of the form $\mathbf{abs_{t}} e$. The label component $\labtrans{\mathit{Ab}}{\cat{AB}}$ has objects $\objects{\cat{AB}}=\idType\rightharpoonup \mathbf{c}$, that map a type name to a type transformer function. The upgrade functor in DSOS just changes this label component, not touching the continuing program code, and the runtime makes sure to use the correct type by applying the type transformer as:
%
\begin{prooftree}
\AxiomC{$\rho_{\mathit{ab}}(\mathbf{t})=\mathbf{c}$}
\UnaryInfC{$\mathbf{abs_{t}} e \transition{\mathit{Ab}=\rho_{\mathit{ab}}\restlabel } \mathbf{c}(e)$}
\DisplayProof\hspace{3ex}
\AxiomC{$\mathbf{t}\not\in\rho_{\mathit{ab}}$}
\UnaryInfC{$\mathbf{abs_{t}} e \transition{\mathit{Ab}=\rho_{\mathit{ab}}\restlabel } e$}
\end{prooftree}

When adding types, the safety check is performed by $\mathsf{updateOK}(-)$ and makes sure that the update information is well typed so that the continuing program will be type safe under the upgraded data. Essentially $\mathsf{updateOK}(-)$ checks that the new type definitions are safe and that the associated type transformers are well typed in the updated type information. It also checks that any new values or function definitions are well typed w.r.t.\ the updated information.

To avoid cluttering more the notation, consider upgrading only type definitions and function declarations, i.e., involve only the pairs categories $\cat{F}$, $\cat{FT}$, $\cat{TY}$, and the discrete category $\cat{AB}$. 
We would define an endo\-functor $E_{\Delta}^{t}$ on $\cat{TY}\prodcat\cat{AB}\prodcat\cat{UTY}$ for updating type definitions, or $E_{\Delta}^{f}$ on $\cat{F}\prodcat\cat{FT}\prodcat\cat{UF}\prodcat\cat{UFT}$ for updating function declarations. The upgrade label categories \cat{UF} and \cat{UFT} contain the same objects as the respective categories, whereas \cat{UTY} maps type identifiers to pairs of a type and a type transformer, as in \proteus.

Consider only $E_{\Delta}^{t}(\rho_{\mathit{ty}},\rho_{\mathit{ab}},\rho_{\mathit{uty}})=$
\[
\phantom{\hspace{1ex}}
\left\{\begin{array}{l@{\hspace{0ex}}r} 
\left(\begin{array}{l}
\rho_{ty}[\identifier{t}\mapsto\sigma \mid \forall \identifier{t}\in\Delta\cap\rho_{uty} \wedge \rho_{\mathit{uty}}(\mathbf{t})=(\sigma,\mathbf{c}) ], \\
\rho_{ab}[\identifier{t}\mapsto \mathbf{c} \mid\forall \identifier{t}\in\Delta\cap\rho_{uty} \wedge \rho_{\mathit{uty}}(\mathbf{t})=(\sigma,\mathbf{c}) ], \\
\rho_{uty}\setminus\Delta, \\
\end{array} \right) & \begin{array}{l} \mbox{ if } \mathsf{updateOK}(-) \\ \end{array} \\
%
&\\
(\rho_{\mathit{ty}},\rho_{\mathit{ab}},\rho_{\mathit{uty}}) & \mbox{otherwise.} 
\end{array} \right.
\]

In the first line we now use the check 
$\mathsf{updateOK}(-)$ as:
\[
\left(\begin{array}{l}
\vdash\rho_{\mathit{ty}}[\rho_{\mathit{uty}}]\ \wedge\ \ \mathit{dom}(\rho_{\mathit{uty}})\in\Delta\ \ \wedge\\
(\forall \mathbf{t}\in\mathit{dom}(\rho_{\mathit{uty}}):\rho_{\mathit{uty}}(\mathbf{t})=(\sigma,\mathbf{c})\Rightarrow\rho_{\mathit{ty}}[\rho_{\mathit{uty}}]\dots\vdash\mathbf{c}:\rho_{\mathit{ty}}(\mathbf{t})\rightarrow\sigma)\ \wedge\\
(\forall \mathbf{f}\in\mathit{dom}(\rho_{\mathit{uf}}): \rho_{\mathit{ty}}[\rho_{\mathit{uty}}]\dots \vdash \rho_{\mathit{uf}}(\mathbf{f}):\rho_{\mathit{uft}}(\mathbf{f}))\\
\end{array} \right)
\]

We have been superficial in the above definition and omitted some details like capabilities and other typing information. To be complete one would use the exact type-and-effect system of \cite[Sec.5]{StoyleHBSN07mutatis}, i.e., from Fig.18-22, and extract the above definition of $\mathsf{updateOK}(-)$ from Fig.23-24. When looking at the definition in \cite[Fig.24]{StoyleHBSN07mutatis} one can correlate the first line above with lines 3-4 (where the $bindOK$ is omitted), the second line with a simplified view of Fig.24(b), and the third line with the rest of Fig.24 that checks the new values.\footnote{Note that the third line of $\mathsf{updateOK}(-)$ would be needed for $E_{\Delta}^{f}$ but not for $E_{\Delta}^{t}$.}
In particular, the $types(H)$ that Fig.24 extracts from the heap, in our case come from the labels like \cat{FT}, which we omitted through ``$\dots$''. 
Useful could be to automate this proof in a proof assistant, on the lines of \cite{pierce12SOS_coq}, which would contain all the meticulous details that have already been done in \cite{StoyleHBSN07mutatis}. 


Considering the same typing system of \cite[Sec.5]{StoyleHBSN07mutatis}, proving type soundness w.r.t.\ the DSOS semantics is not more than redoing the lengthy details from the appendix of \cite{StoyleHBSN07mutatis}.
The statement in Proposition~\ref{prop_type_sound} reflects the DSOS style, but can easily be matched by the respective statement in \cite[Th.A.22]{StoyleHBSN07mutatis}. This is a specific result for the language of \cite{StoyleHBSN07mutatis}, meant here for exemplification, and not an essential part of the DSOS framework. Therefore, we only outline how the proof would go, which is following the standard method for such type soundness proofs \cite{FelleisenWright94TypeSound,AghaMST97actorFoundation,abadi2012theory}.

\begin{proposition}[type soundness]\label{prop_type_sound}
For a program term $P$ and an object $o$ from the label category used in the semantics we have that if for some type environment $\Omega$,

\[
\Omega \vdash P:\sigma,\Omega'\mbox{\ \ and\ \ } \Omega\vdash o
\]
then either $P$ is a value, or there exists a transition $P\transition{\alpha}P'$, with $o=\source{\alpha}$, $o'=\target{\alpha}$, for which $\Omega'\vdash o'$ and $\Omega'\vdash P':\sigma,\Omega''$, where $\Omega',\Omega''$ are the effects of the typing judgments containing generated typing information.
Particularly interesting is the above statement with an empty type environment and the object containing only empty maps.
\end{proposition}

\begin{proof}[sketch]
The check $\Omega\vdash o$ corresponds to the check that the heap is well typed in \proteusSrcLang. The program $P$ and the object $o$ together make up the configuration that is used in \proteusSrcLang. When the code is not a value, it can reduce to a new program of the same type and a changed heap which is still well typed. For upgrades this ensures well typedness of the changed heap.
\end{proof}

\section{Encapsulating MSOS for object-oriented languages}\label{sec_ex_DSOS_creol}\label{sec_encapsulation}

We will show how to give semantics in a modular style to concurrent object-oriented constructs as used by the language \creol. For this we first need to define a new \textit{encapsulating mechanism} for concurrent object-orientation. This construction extends MSOS in a conservative manner and upholding the modularity principles as explained in the end. The construct is not specific to object-orientation, but can be applied to other programming settings where execution is encapsulated in some way, e.g., where one talks about isolating execution environments like in ambient calculus \cite{cardelli2000mobile} or distributed settings \cite{hennessy2002resource,hennessy2007distributed}.


We focus here on the concurrency notion from the Actor model \cite{AghaMST97actorFoundation} which has proved well suited for the object-oriented languages. In this setting concurrent objects communicate through asynchronous method calls and have their own execution unit (like a virtual CPU), thus having standard programming constructs be run \textit{inside} the object. 
This notion of encapsulation of the execution must be captured in the category theory of the labels. We provide for this an \textit{encapsulating construction}. The term ``encapsulate'' has a specific meaning in object-oriented languages. Our categorical construction has a similar intuition, therefore we prefer the same terminology.

Not only the code is encapsulated in an object, but also the auxiliary data that is used to give semantics to the code. These data components are now private to the specific object. We want to keep the modularity in defining semantics for object-oriented constructs. We want that definitions of new semantic rules would not change the definitions of the old rules. On the contrary, we may use the old transition relation to define new transition relations. Essentially, we will encapsulate old transitions into transitions that are localized to one object. In the concurrent setting, we even see how more objects may perform transitions localized to each of them, thus making a global transition, changing many of the local data.

\vspace{2ex}
\noindent\begin{minipage}[l]{0.68\textwidth}
\begin{definition}[natural transformations]\label{def_nattrans}
Consider two arbitrary categories \cat{A} and \cat{B} and two functors $F,G$ from \cat{A} to \cat{B}. A \emph{natural transformation} $\naturalTransf:F\rightarrow G$, from the functor $F$ to $G$, is defined as a function that associates to each object $o$ of $\objects{\cat{A}}$ a morphism $\beta$ of $\morphisms{\cat{B}}$ with $\source{\beta}=F(o)$ and $\target{\beta}=G(o)$ s.t.\ for any morphism $\alpha$ of $\morphisms{\cat{A}}$, with $\source{\alpha}=o$, the diagram on the right commutes.
\end{definition}
\end{minipage}
\hspace{2ex}\begin{minipage}[r]{0.3\textwidth}
\begin{diagram}
$F(o)$ & $\rTo^{\naturalTransf(o)}$ & $G(o)$ \\
$\dTo_{F(\alpha)}$ &  & $\dTo_{G(\alpha)}$ \\
$F(\ensuremath{o'})$ & $\rTo^{\naturalTransf(o')}$ & $G(\ensuremath{o'})$ \\
\end{diagram}
\vfill
\end{minipage}

\begin{definition}[encapsulating construction]\label{def_encapsulate}
Let \cat{O} be a discrete category, and \cat{A} a label category.
The \emph{encapsulating construction} $\encapsulate{\cat{O}}{\cat{A}}$ returns a category \cat{E} with all the functors $F:\cat{O}\rightarrow\cat{A}$ as objects, and natural transformations between these functors as morphisms.
\end{definition}

The discrete category \cat{O} captures programming objects identifiers (i.e., each object of the category is a unique identifier for a programming object). 
Other categories may be used if one needs to capture relations between the programming objects, like ownership.
The intuition is that each functor attaches to each programming object identifier one data object from \cat{A}, thus capturing one snapshot of the working data of all programming objects in the system. We have access to these pieces of working data by taking the appropriate identifier, i.e., $F(o)$ is the data encapsulated in the programming object identified by $o$. 

A morphism in $\encapsulate{\cat{O}}{\cat{A}}$, i.e., a natural transformation, between two such snapshots $F,F'$ can be thought as capturing one way of transforming one snapshot into the other.
These intuitions hold also when monoid categories are part of the labels. In this case there are multiple natural transformations between two functors.

\begin{notation}
We can either write $\eta$ as a pair of functors $(F,F')$ or we can write it as a set of morphisms from \cat{A} indexed by the objects $o\in \cat{O}$, i.e., $\eta=\{(F(o),F'(o))\mid o\in\objects{\cat{O}}\}$.
As such we may refer to the data morphisms from the encapsulated label category \cat{A}, since these are indexed by the programming object identifiers, i.e., $\eta(o)$ and call these ``local'' morphisms associated to $o$. 
In consequence, we are free to use the get operation to refer to a particular component of the encapsulated label category morphisms $\eta(o)$, i.e., we may write $\eta(o).i$ or any other preferred notation like $o.i$ or $o\mapsto i$ or $\langle o \mid i\rangle$ or $o:i$.
\end{notation}

One property of the encapsulation construction is that the resulting category is similar to the encapsulated category in the following sense.

\begin{proposition}\label{prop_enc1}
When the encapsulating construction is applied to a label category \cat{A} where the morphisms are uniquely defined by the objects (i.e., with properties as in Proposition~\ref{prop_propertiesDiscretPair}, e.g., a discrete or pairs category), then the morphisms of\, $\cat{E}=\encapsulate{\cat{O}}{\cat{A}}$ are uniquely defined by the objects (i.e., functors).
\end{proposition}

\begin{proof}

\noindent\begin{minipage}[l]{0.65\textwidth}
\setlength\parindent{14pt}
The objects of \cat{E} are functors $F:\cat{O}\rightarrow\cat{A}$. Take two such functors $F,F'$; a morphism between them is a natural transformation $\eta$ which for each object of \cat{O} associates one morphism of \cat{A}, i.e., $\eta(o)\in\morphisms{\cat{A}}$, with the following property: for some $o\in\objects{\cat{O}}$ and some morphism $\alpha\in\morphisms{\cat{O}}$ with source $o$ and target $o'$, the diagram on the right commutes.

In our case this diagram becomes simpler because in \cat{O} the only
morphisms are the identities, which means that $\alpha$ is in fact
$id_{o}$ and thus the $o'$ in the diagram above is just $o$. Moreover,
the functors take identities to identities, so $F(\alpha)$ becomes
$id_{F(o)}$. Then the diagram becomes the one to the right, which
clearly commutes for any $\eta$.

The natural transformation $\eta$ assigns the morphism $\eta(o)$
between $F(o)$ and $F'(o)$ in \cat{A}, which is unique by the
assumption that in \cat{A} morphisms are uniquely determined by the
objects on which they act, i.e., $\eta(o)=(F(o),F'(o))$. The same for
any $o'\in\objects{\cat{O}}$ the $\eta(o')$ is unique.
In consequence, the $\eta$ is uniquely defined by the two functors on which it is applied.
%
\end{minipage} 
\begin{minipage}[r]{0.35\textwidth}
\vspace{-5ex}\begin{diagram}
$F(o)$ & $\rTo^{\eta(o)}$ & $F\ensuremath{'}(o)$ \\
$\dTo_{F(\alpha)}$ &  & $\dTo_{F'(\alpha)}$ \\
$F(\ensuremath{o'})$ & $\rTo^{\eta(o')}$ & $F\ensuremath{'}(\ensuremath{o'})$ \\
\end{diagram}
\begin{diagram}
$F(o)$ & $\rTo^{\eta(o)}$ & $F\ensuremath{'}(o)$ \\
$\dTo_{id_{F(o)}}$ &  & $\dTo_{id_{F'(o)}}$ \\
$F(o)$ & $\rTo^{\eta(o)}$ & $F\ensuremath{'}(o)$ \\
\end{diagram}
\end{minipage} 

 \end{proof}

The category built by the encapsulating construction can be used with the label transformer to attach more global data structures.
Therefore, the encapsulating construction is modular, in the spirit of MSOS, in the sense that new global programming constructs and rules may be added without changing the rules for encapsulation. 
The reference mechanism provided by the label transformer is used as normal. We see this in Subsection~\ref{example_asynchronous_method_calls} on asynchronous method calls where additional global structures are needed for keeping track of the messages being passed around.

Moreover, we may encapsulate this category again, wrt.\ a new discrete category, giving a different set of identifiers. This has application in languages with object groups, like ABS \cite{johnsen16scp}, where objects execute inside a group.

The encapsulating construction preserves modularity also in the sense that new programming constructs may be added to run localized (inside objects), and thus the encapsulated category may need to be extended to include new auxiliary data components. The encapsulation is not affected, in the sense that the rules for encapsulation, or rules that were defined referring to some encapsulated data, need no change. 
The reference mechanism (with the get operation provided by the label transformer) used in defining the localized rules is independent of the new local categories added. 
This aspect becomes apparent when treating \textit{threads} in Subsection~\ref{example_threads}.
Henceforth we denote the encapsulated (or local or internal) category by \cat{I} when its components are irrelevant.

The way of applying the encapsulating construction will use transitions labelled both with morphisms from \cat{I} as well as from $\encapsulate{\cat{O}}{\cat{I}}$, which will not fit the MSOS type of transition systems. As such we define a slight extension in Definition~\ref{def_EncALTS}. 

\begin{definition}[Encapsulated ALTS]\label{def_EncALTS}
For a set of categories $\{\cat{A}_{i}\}$ define an \emph{encapsulating arrow-labelled transition system}
$(\Gamma,\bigcup_{i}\morphisms{\cat{A}_{i}},\transition{})$ formed
by 
a set of
\emph{states} $t\!\in\!\Gamma$,
including an \emph{initial state} $t_0$, and transitions  $\transition{\alpha}$ labelled by morphisms $\alpha$ from one of the categories $\cat{A}_{i}$.
A \emph{com\-putation} in an encapsulated ALTS is a sequence $t_{0}\!\transition{\alpha_{0}}\!t_{1}\!\transition{\alpha_{1}}\!t_{2}\dots$ s.t.\ for any $t_{i}\!\transition{\alpha_{i}}t_{i+1}\!\transition{\alpha_{i+1}}t_{i+2}$ the two morphisms are both coming from the same category $\cat{A}_{i}$ and are composable in $\cat{A}_{i}$ as $\alpha_{i+1}\circ\alpha_{i}\in\!\morphisms{\cat{A}_{i}}$.
\end{definition}


\subsection{Modular SOS for concurrent object-orientation}\label{subsec_MSOS_creol}

The encapsulating construction is used to give semantics to concurrent object-oriented programming languages where code is executed locally, in each object, and the objects are running in parallel, maybe communicating with each other. The modularity is obtained by defining the localized transitions in terms of the transitions defined for the individual executing programming constructs, as given by the rule in Subsection~\ref{example_objects}.

\begin{notation}
We reuse and extend the notation from Section~\ref{sec_MSOS_proteus} to specify (partly) the morphisms on arrows of encapsulated ALTS in the rules below. In particular, when specifying encapsulating morphisms (i.e., natural transformations from $\encapsulate{\cat{O}}{\cat{I}}$) we use the notation $\mathbf{o}:X$ to partly specify the natural transformation, saying that the specific programming object $\mathbf{o}$ has in the encapsulated category the morphism $X$ (whichever that is). Similarly, the use of \dots\ around this (i.e., at the level of the category $\encapsulate{\cat{O}}{\cat{I}}$) means that the rest of the morphisms from the natural transformation would be identify morphisms (i.e., for the other programming object identifiers). We let $X$ stand for an arbitrary morphism also in the $\encapsulate{\cat{O}}{\cat{I}}$ when this is clear from the context.
\end{notation}

\subsubsection{Objects}\label{example_objects}
We add \textit{object identifiers} as constants denoted $\mathbf{o}\in\idObj$. 
We add one programming construct of a new sort called \emph{Objects}, denoted $O$, which localizes a term of sort statement wrt.\ an object identifier. 
%
\[
O\ \ ::=\ \ \objectTerm{o}{s}\ 
\]
This signature $\SigmaExec_{\ref{example_objects}}$ should include some signature defining statements; any of the constructs before can run inside the object construction, but the exact set of constructs is not relevant for the transition rules below. 

The semantics of \textit{object programs} is given using transitions labelled from a category constructed using the encapsulating construction applied to some appropriate \cat{I}: $\cat{E}=\encapsulate{\cat{O}}{\cat{I}}$, where $\objects{\cat{O}}=\idObj$. 
Since any of the constructs before can be run inside the object construction, therefore we encapsulate the category that we built before. 
Thus, the label category that we use in the rules for the object construction below would be
\[                                                                                                                                                                                                                                                                                         \encapsulate{\cat{O}}{\labtrans{F}{\cat{F}}(\labtrans{S}{\cat{S}}(\labtrans{R}{\cat{R}}(\trivcat)))}.
\]

We give one transition rule that \textit{encapsulates} any transition at the level of the statements inside the objects.
\begin{prooftree}
\AxiomC{$s\transition{X}s'$}
\RightLabel{\textsc{(enc)}}
\UnaryInfC{$\objectTerm{o}{s} \transition{[\mathbf{o}:X\dots]} \objectTerm{o}{s'}$}
\end{prooftree}
The label $X$ stands, as before, for any morphism in the local category \cat{I}. The label of the conclusion is taken as a morphism in the encapsulation category \cat{E}. The notation $[\mathbf{o}:X\dots]$ specifies only part of the natural transformation, whereas the rest may be any identity morphism. This specifies that the data for the object $\mathbf{o}$ is known before and after the local execution, whereas the local data of any other objects are irrelevant and may be anything, but is not changed in any way. Therefore, any functors $F,F'$ that respect the fact that they assign to $\mathbf{o}$ the source and target objects of $X$, and may assign anything to all other objects, are good. Moreover, the monoid labels that may appear in $X$ are part of the specific natural transformation that we choose between the two functors $F,F'$; i.e., it is exactly the natural transformation assigning to $\mathbf{o}$ the morphism $X\in\morphisms{\cat{I}}$.

\subsubsection{Systems of objects}\label{example_systems_of_objects}
Objects may run in parallel, thus forming systems of distributed objects. For this we add a parallel construct $\parallel$ of sort \emph{Objects}, with all object identifiers different:
\[
O\ \ ::=\ \ \mathit{obj}_{1}\parallel\mathit{obj}_{2}\ (\mathit{obj}_{1},\mathit{obj}_{2}\in O) \mid \dots
\]
We choose an \textit{interleaving} semantics for our parallel operator, hence the rules:
\begin{prooftree}
\hspace{-5ex}\AxiomC{$\mathit{obj}_{1}\transition{X}\mathit{obj}'_{1}$}
\RightLabel{\textsc{(int-1)}}
\UnaryInfC{$\mathit{obj}_{1}\parallel\mathit{obj}_{2} \transition{X} \mathit{obj}'_{1}\parallel\mathit{obj}_{2}$}
\DisplayProof\hspace{1ex}
\AxiomC{$\mathit{obj}_{2}\transition{X}\mathit{obj}'_{2}$}
\RightLabel{\textsc{(int-2)}}
\UnaryInfC{$\mathit{obj}_{1}\parallel\mathit{obj}_{2} \transition{X} \mathit{obj}_{1}\parallel\mathit{obj}'_{2}$}
\end{prooftree}
Note that the $X$ in this rule stands for any morphism in the encapsulating category, whereas in the previous rule it was standing for morphisms in the local category.


We may easily specify non-interleaving concurrency by specifying more precisely the label components:
\begin{prooftree}
\AxiomC{$\objectTerm{o_{1}}{s_{1}}\transition{\mathbf{o_{1}}:X}\objectTerm{o_{1}}{s'_{1}}$}
\AxiomC{$\mathit{obj}_{2}\transition{\naturalTransf}\mathit{obj}'_{2}$}
\RightLabel{\textsc{(non-int)}}
\BinaryInfC{$\objectTerm{o_{1}}{s_{1}}\parallel\mathit{obj}_{2} \transition{\naturalTransf[\identifier{o_{1}}:X]} \objectTerm{o_{1}}{s'_{1}}\parallel\mathit{obj}'_{2}$}
\end{prooftree}
The label of the conclusion specifies the morphism which is the natural transformation \naturalTransf\ changed so that it incorporates the specified local morphism of $\identifier{o_{1}}$. In this way any number of objects may execute local code and the local changes to their data is visible in the global label.

\subsubsection{Methods inside objects}\label{example_methods_Appendix}

Methods are like functions only that they have a \textbf{return} statement which is treated specially.\footnote{Other programming options are possible like having functions evaluate to a value, and thus not use the return statement.} We thus add method definition and invocation as $\SigmaExec_{\ref{example_methods_Appendix}}$:
\[
d\ \ ::=\ \ \MethodDef{m}{s} \mid \dots
\]
\[
s\ \ ::=\ \ \return{e} \mid y:=\ \identifier{m}(e) \mid \dots
\]
For simplicity we limit the discussion to methods with one input and one output.
The transition rules for methods use another pairs label category \cat{MD} for storing method definitions (the same as was done for function definitions) which is added by the label transformer, identified by the index $MD$, to the local labels category \cat{I} that is encapsulated.
In order to define the semantics of the (local) call statement $y:=\identifier{m}(e)$ in isolation, we introduce an additional construct $\methodend{y}$ to control the passing of the return value to the actual output variable of the call, i.e., $y$.
An alternative would be to use the let construct to bind the return value
in the statements following the call, however, that would require 
identification of these statements in the rules.
\begin{prooftree}
\AxiomC{\phantom{$\identifier{m}\not\in\rho_{m}$}}
\UnaryInfC{$\MethodDef{m}{s} \transition{\{MD=\rho_{m}\restlabel MD=\rho_{m}[\identifier{m}\mapsto \lambda(\mathbf{x}).(s)]\}} \nilValue$}
\end{prooftree}

\begin{prooftree}
\hspace{-3ex}\AxiomC{$e\transition{X}e'$}
\UnaryInfC{$\return{e} \transition{X} \return{e'}$}
\DisplayProof\hspace{3ex}
\AxiomC{$e\transition{X}e'$}
\UnaryInfC{$y:=\identifier{m}(e) \transition{X} y:=\identifier{m}(e')$}
\DisplayProof\hspace{3ex}\vspace{3mm}
\AxiomC{$\rho_{m}(\identifier{m})=\lambda(\mathbf{x}).(s)$}
\UnaryInfC{$y:= \identifier{m}(v) \transition{\{MD=\rho_{m}\restlabel\}}
   (s)[v/\mathbf{x}] \methodend{y} $}
\end{prooftree}
%
%

\begin{prooftree}
\AxiomC{\phantom{
}}
\UnaryInfC{$\return{v};s \methodend{y}
  \transition{\{
\restlabel\}}
  y:=v$}
\end{prooftree}
In the last rule, $s$ is any  statement list; thus it does not contain the special 
$\methodend{\ldots}$ construct, which is not a regular statement.
This  ensures that (local) calls are handled in a stack-based manner.
We do not complicate the presentation more because our aim is only to exemplify how the \creol\ language can be given a MSOS style of semantics, using the encapsulation construction.

\subsubsection{Threads}\label{example_threads}
We take the model of threads studied in 
\cite{AbadiP09popl,abadiplotkin10LMCS}
and consider the following programming constructs of sort statement in a signature $\SigmaExec_{\ref{example_threads}}$ which normally would include also other constructs for statements from before:
\[
s\ \ ::=\ \ \yieldThread \mid 
\asyncThread{s}\mid \dots
\]

Threads need an additional data component called \textit{thread pool}. We build a pairs category \cat{T} which has as objects thread pools. 
The internal label category \cat{I} (chosen depending on the other constructs) is extended with $\labtrans{T}{\cat{T}}$.
The label category used to give the transition rules for statements becomes now:
\[
\labtrans{T}{\cat{T}}(\labtrans{R}{\cat{R}}(\labtrans{F}{\cat{F}}(\labtrans{S}{\cat{S}}(\trivcat)))).
\]

We need more algebraic structure for the thread pools, which is used when defining the transition rules.
A thread pool may be implemented in multiple ways (e.g., as sets or lists); here we only require two operations on a thread pool, an \textit{insertion} \insertThread\ and a \textit{deletion} \deleteThread\ operation. Take $\rho_{t}$ to be a thread pool and $s$ a program term, then $\rho_{t}\insertThread s$ is also a thread pool containing $s$; and when $s\in\rho_{t}$ then $\rho_{t}\deleteThread s$ is also a thread pool that is the same as $\rho_{t}$ but does not contain $s$.

Because of the \yieldThread, which needs the whole program term that follows it, we give semantics to threads using \textit{evaluation contexts}. The MSOS is perfectly suited for describing semantics using evaluation contexts. 
One may define rules for a programming construct both using evaluation contexts and without; and then pick the preferred rules. An essential result is to show that both sets of rules generate the same arrow-labelled transition system.

Evaluation contexts are statements with a \textit{hole} \holeContext{\,}: 
\[
\evalContext\ \ ::=\ \ \holeContext{\,} \mid \evalContext \sequence s
\]

Placing a program term $s$ in the whole of a context \evalContext\ is denoted \evalContextFill{s} and results in a normal program term (i.e., without the hole).
It is essential to prove that any statement in the language can be \textit{uniquely} decomposed into an evaluation context \evalContext\ and a program term $s$ so that the choice of transition rules is unambiguous. For the simple contexts that we defined above, this result is easy.

Instead of giving alternative rules using evaluation contexts, 
we prefer to give the following rule, and remove the two rules for sequential composition from Subsection~\ref{example_noLabels}.
A second rule is required when object terms are present. The $X$ label on the left comes from an encapsulated \cat{I}, whereas the one on the right comes from a global label.
\begin{prooftree}
\AxiomC{$s\neq\nilValue$}
\AxiomC{$s\transition{X}s'$}
\BinaryInfC{$\evalContextFill{s}\transition{X} \evalContextFill{s'}$}
\DisplayProof\hspace{3ex}
\AxiomC{$s\neq\nilValue$}
\AxiomC{$\objectTerm{o}{s}\transition{X}\objectTerm{o}{s'}$}
\BinaryInfC{$\objectTerm{o}{\evalContextFill{s}}\transition{X} \objectTerm{o}{\evalContextFill{s'}}$}
\end{prooftree}

Now we can give the rules for the new programming constructs, which may be compared to the ones given in 
\cite[Fig.4]{AbadiP09popl}.
%
\begin{prooftree}
\AxiomC{\phantom{$s\neq\nilValue $}}
\UnaryInfC{$\asyncThread{s}\transition{\{T=\rho_{t}\restlabel T=\rho_{t}\insertThread s\}} \nilValue$}
\DisplayProof\hspace{3ex}
\AxiomC{\phantom{$s\neq\nilValue$}}
\UnaryInfC{$\evalContextFill{\yieldThread}\transition{\{T=\rho_{t}\restlabel T=\rho_{t}\insertThread \evalContextFill{
\nilValue
}\}} \nilValue$}
\end{prooftree}

\begin{prooftree}
\AxiomC{
$s\in\rho_{t}$}
\UnaryInfC{$\nilValue\transition{\{T=\rho_{t}\restlabel T=\rho_{t}\deleteThread s\}} s$}
\end{prooftree}

\subsubsection{Classes}\label{example_classes}

It is common in the setting of object-orientation to have
\textit{method definitions} part of \textit{class definitions}, where
objects are instances of such classes and can be created anytime with
the \newObjectOf{} programming construct. Inheritance and interfaces
are normally part of class definitions, 
but are not essential here;
these can be easily added as in \cite{JO05inheritance}.

\textit{Class identifiers} are introduced from a set $\idClass$,
usually written as \identifier{C}. Class definitions include
\textit{method definitions} and 
an \textit{intialization} with 
initialized \textit{attribute definitions} and initial statements;
%
\[
\AttributesSort\ \ ::=\ \ 
s
\hspace{6ex}
M\ \ ::=\ \ \MethodDef{m}{s}\ \mid M \sequence M
\]
\[
d\ \ ::=\ \ \classDefNameMethods{C}{\AttributesSort\sequence M} \mid \dots
\hspace{5ex}
s\ \ ::=\ \ \assignment{x}{\newObjectOf{\identifier{C}}} \mid \identifier{m}(e)\  \mid \dots
\]

For the semantics we need two global category components (i.e., not local to the objects) which keep definitions of methods for each class and another to keep the attributes. Denote these by \cat{C} and \cat{A}, and associate using the label transformer the indexes $C$ and $A$. 
The objects $\rho_{c}\in\objects{\cat{C}}$ are mappings from class identifiers to definitions of methods; i.e., $\rho_{c}:\idClass\rightharpoonup(\idMethods\rightharpoonup\mathit{MtdDef})$.
Objects $\rho_{a}\in\objects{\cat{A}}$ are mappings $\idClass\rightharpoonup \AttributesSort$.
The encapsulation is a global component of its own, to which the label transformer associates index $E$. 
The transition rule for class definitions is:
\begin{prooftree}
\AxiomC{
$\rho'_{a}=\rho_{a}[\identifier{C}\mapsto \AttributesSort]$}
\AxiomC{$\rho'_{c}=\rho_{c}[\identifier{C}\mapsto \{\identifier{m}\mapsto\lambda(\mathbf{x}).(s)\mid\identifier{m}\in M\}]$}
\BinaryInfC{$
{\classDefNameMethods{C}{\AttributesSort\sequence M}}
\transition{\{A=\rho_{a},C=\rho_{c}\restlabel
  A=\rho'_{a},C=\rho'_{c}\}} 
{\nilValue}$}

\end{prooftree}

Each object is an instance of a class. In consequence we associate to
each object the name of the class it belongs to, and from where method
definitions can be retrieved.\footnote{This is the
  \textit{dynamic binding} notion (also known as late binding, or
  dynamic dispatch) where the method definitions are retrieved when
  they are needed. This is especially useful in the presence of
  inheritance and dynamic class upgrades, as in
  Section~\ref{subsec_ex_DSOS_creol}; otherwise we could do without,
  and use the method definitions local to objects as in
  Subsection~\ref{example_methods_Appendix}. Normally this class
name information is held in a special variable of the object, but here
we will use a category component, to keep with the modular
style.} 
Therefore, to the internal category \cat{I} we add one more
category \cat{CN}, to which the label transformer will associate the
index $CN$.  The objects $\objects{\cat{CN}}=\idClass$ are just class
identifiers.\footnote{The objects of this category have such a simple structure
  that it may look awkward to have a category defined on them,
  but it is perfectly fine for MSOS and encouraged for separation of concerns (not optimisation).}
The rule for object creation is:
\begin{prooftree}
\AxiomC{
$\mathit{fresh}(\identifier{o'})$}
\AxiomC{$\forall i\neq CN\ \ \identifier{o'}:i=\emptyset$}
\AxiomC{$\rho_{a}(\identifier{C})=\AttributesSort$}
\TrinaryInfC{$\objectTerm{\!o\!\!}{\!\!\evalContextFill{\assignment{x\!}{\newObjectOf{\identifier{C}}}}\!}\!\!\transition{\{A=\rho_{a},C =\rho_{c},\identifier{o}:S=\rho \restlabel \identifier{o'}:CN=\identifier{C},\identifier{o}:S=\rho[\identifier{x}\mapsto \identifier{o'}]\}}\!\!\objectTerm{\!o\!\!}{\!\!\evalContextFill{\nilValue}\!}\!\!\parallel\!\!\objectTerm{\!o'\!\!}{\!\!\AttributesSort\!}$}
\end{prooftree}

There are different ways of ensuring freshness of the object identifiers and different ways of 
initialising generated objects, for instance by means of constructors.
Our notion allows initialised attribute declarations
as well as initial statements, for instance a call to a local method
(which is used to start desired active behaviour in the case of
\creol).
Due to the assumption of distinct variables names,
we do not need to separate attributes and local variables.

The transition rule for method application must include the object because it needs the global class definitions where the method definitions are found. 
\begin{prooftree}
\AxiomC{$e\transition{X}e'$}
\UnaryInfC{$y:=\identifier{m}(e) \transition{X} y:=\identifier{m}(e')$}
\DisplayProof\hspace{3ex}
\AxiomC{$\identifier{C}\in\rho_{c}$}
\AxiomC{$\identifier{m}\in\rho_{c}(\identifier{C})$}
\AxiomC{$\rho_{c}(\identifier{C})(\mathbf{m})=\lambda(\mathbf{x}).(s)$}
\TrinaryInfC{$\objectTerm{o}{y:=\identifier{m}(v)}
  \transition{\{\identifier{o}:CN=\identifier{C},C=\rho_{c}
    \restlabel\}} \objectTerm{o}{s[v/\mathbf{x}] \methodend{y}}$}
\end{prooftree}


\subsubsection{Asynchronous method calls as in \creol}\label{example_asynchronous_method_calls}
We take the model of asynchronous method calls from \cite{JohnsenO07} and consider two programming constructs for calling a method and reading the result of the completion of a call:
\[
s\ \ ::=\ \ \callMethodOfIn{\identifier{m}(e)}{\identifier{o}}{\identifier{t}} \mid \readMethodReturnFromIn{\identifier{t}}{\identifier{x}} \mid \return{e} \mid \dots
\]
where $\identifier{t}\in\idFutures$ are special identifiers used for retrieving the result of the method call.
This mechanism has been studied as ``futures'' in programming
languages \cite{Halstead85multilisp,FlanaganF99futures,deboer07esop}.

Denote this signature $\SigmaExec_{\ref{example_asynchronous_method_calls}}$, which can be added to any previous signature. 

The asynchronous method calls, as discussed in \cite{JohnsenO07}, work
with asynchronous message passing, as in the Actor model
\cite{AghaMST97actorFoundation}. In consequence we need a global data
component to keep track of the messages in the system. We consider
each object having a pool of messages. Since the message pools will be
manipulated by the distributed objects of the system we use a pairs
category \cat{M} with objects $\objects{\cat{M}}=\idObj\rightarrow
2^{\mathit{MsgTerm}}$ being mappings from object identifiers to
message sets. The label transformer \labtrans{M}{\cat{M}} is applied
at least to an encapsulating category.  Similarly to the thread pools,
define set operations \insertThread\ and \deleteThread\ to add and
remove messages from any set $\mathcal{MS}\in
2^{\mathit{MsgTerm}}$. For our exemplification purposes the messages
are of the form: \hspace{2ex}
$\invoke{\identifier{o},n,\identifier{m}(v)}$ \ \ and\ \
$\completion{n,v}$, where $\identifier{o}$ is an object identifier,
$n\in\mathbb{N}$ is a natural number (representing the \textit{future} of
  the call), and $\identifier{m}(v)$ represents the method named
$\identifier{m}$ and $v$ a value term.  Because of the asynchronous
method calling scheme, the method declarations are particular in the
sense that the first two parameters are predefined for all methods as
being \caller\ and \futureLabel, 
and the statements may end with a
return statement: 
$\MethodDefAsync{m}{s}{e}$.
We here consider \emph{local futures}
  as opposed to shared futures. The latter would require the presence
  of special future objects, allowing multiple reads.

The special future identifiers $\identifier{t}$ can be seen as variables which may hold only natural numbers and cannot be modified by the program constructs, but only by the semantic rules. Since identifiers \identifier{t} are local to the objects, we extend the category \cat{I} by attaching another data component \labtrans{L}{\cat{L}}. The category \cat{L} is a pairs category with objects $\objects{\cat{L}}$ being mappings $\idFutures\rightharpoonup \mathbf{Nat}$.
\begin{prooftree}
\AxiomC{$\mathit{fresh}(n,\rho)$}
\AxiomC{$\rho_{m}(\identifier{o'})=\mathcal{MS}$}
\AxiomC{$o'\not = o$}
\TrinaryInfC{$\objectTerm{\!o\!\!}{\!\!\callMethodOfIn{\identifier{m}(v)}{\identifier{o'}}{\identifier{t}}} \transition{\!\!\identifier{o}:L=\rho,M=\rho_{m}\restlabel M=\rho_{m}[\identifier{o'}\mapsto\mathcal{MS}\insertThread \invoke{\identifier{o},n,\identifier{m}(v)}],\identifier{o}:L=\rho[\identifier{t}\mapsto n]\!\!} \objectTerm{\!o\!\!}{\!\!\nilValue}$}
\end{prooftree}
%
\begin{prooftree}
\AxiomC{$\rho_{m}(\identifier{o})=\mathcal{MS}$}
\AxiomC{$\invoke{\identifier{o'},n,\identifier{m}(v)}\in\mathcal{MS}$}
\BinaryInfC{$\objectTerm{o}{s} \transition{M=\rho_{m}\restlabel M=\rho_{m}[\identifier{o}\mapsto\mathcal{MS}\deleteThread \invoke{\identifier{o'},n,\identifier{m}(v)}]} \objectTerm{o}{\asyncThread{\identifier{m}(\identifier{o'},n,v)}\sequence s}$}
\end{prooftree}
%
\begin{prooftree}
\AxiomC{$\rho(\caller)=\identifier{o'}$}
\AxiomC{$\rho(\futureLabel)=n$}
\AxiomC{$\rho_{m}(\identifier{o'})=\mathcal{MS}$}
\TrinaryInfC{$\objectTerm{o}{\evalContextFill{\return{v}}} \transition{\identifier{o}:S=\rho,M=\rho_{m}\restlabel M=\rho_{m}[\identifier{o'}\mapsto\mathcal{MS}\insertThread \completion{n,v}]} \objectTerm{o}{\nilValue}$}
\end{prooftree}
%
\begin{prooftree}
\AxiomC{$\rho(\identifier{t})=n$}
\AxiomC{$\rho_{m}(\identifier{o})=\mathcal{MS}$}
\AxiomC{$\completion{n,v}\in\mathcal{MS}$}
\TrinaryInfC{$\objectTerm{o}{\readMethodReturnFromIn{\identifier{t}}{\identifier{x}}} \transition{\identifier{o}:L=\rho,M=\rho_{m}\restlabel M=\rho_{m}[\identifier{o}\mapsto\mathcal{MS}\deleteThread \completion{n,v}]} \objectTerm{o}{\assignment{x}{v}}$}
\end{prooftree}
%
Essential to the above rules
is that in each rule only one object term is present, thus capturing
the asynchronous method call aspect. Moreover, one can clearly see the
production and consumption of the messages.  The freshness of $n$ in
$\rho$, that is required in the first rule, can be obtained in various
ways, which only complicate rules, and we decide to leave these
details out of this presentation.

\begin{remark}
Rules two and four are dependent on additional program constructions, and thus on their semantics. This is not in the modular spirit. We would achieve the same effect by simulating the two corresponding transition rules (for \textbf{async} and assignment) and modify the required local data components directly in the rule above; this means that the second rule would involve the \cat{T} local category and the last rule would involve \cat{S}. In this way dependency on program constructs is removed, but still the rules depend on the two local label components. This is more preferred in the modular SOS.
\end{remark}

\intechrep{If in the second rule we use \nilValue\ instead of s we would restrict the point where messages can be treated. This is the same as the point where new threads can be started. There would be a race between the two rules for threads and messages, and we would have to solve this race.}

There are several variations on giving semantics to asynchronous method calls; the above is just our choice. Other choices can be to have a global store where the values that are returned by the call are kept and retrieved by the caller (not using the completion message as we do above). Other choices do not necessarily block on a read, as we do in the last rule above, but put the waiting process in the thread pool.


\section{Exemplifying Dynamic SOS for \creol}\label{subsec_ex_DSOS_creol}

First we identify the data components that are subject to the dynamic upgrade.
For \creol\ our example upgrades classes that have only methods and attributes.
Thus, the data components subject to the upgrade are \cat{C} and \cat{A} holding the methods respectively attributes for each class.
In \cite{JohnsenOS05dynamicClasses} extra complexity appears in the form of \textit{dependencies between upgrades}. In consequence, in \cite{JohnsenOS05dynamicClasses} classes have associated \textit{upgrade numbers}, that are only inspected by the objects during method calls, and changed only by the upgrade constructs. A discrete category \cat{UN}, with objects $\objects{\cat{UN}}=\idClass\rightharpoonup\mathbf{Nat}$, mappings from class identifiers to natural numbers, is added as a global component \labtrans{\mathit{UN}}{\cat{UN}}. 
Denote the product of all these data categories as $\cat{D}=\cat{C}\prodcat\cat{A}\prodcat\cat{UN}$.

Next we identify the upgrade information, looking at \cite{JohnsenOS05dynamicClasses}, as three components: two holding the actual new code for methods and attributes, and another holding the dependencies, i.e.,
\begin{itemize}
\item  a discrete category \cat{UC} with the same objects as \cat{C}, $\objects{\cat{UC}}=\idClass\rightharpoonup(\idMethods\rightharpoonup\mathit{MtdDef})$, holding information about which class names need to be upgraded and what is the new information to be used;
\item  another discrete category \cat{UA} has objects $\idClass\rightharpoonup A$;
\item  and another \cat{UD} having objects $\objects{\cat{UD}}=\idClass\rightharpoonup(\idClass\rightharpoonup\mathbf{Nat})$ holding upgrade information about which class depends on which versions of which classes.
\end{itemize}
Denote the upgrade categories as $\cat{U}_{\cat{D}}=\cat{UC}\prodcat\cat{UA}\prodcat\cat{UD}$.
Thus, the endo\-functors are defined on $\cat{D}\prodcat\cat{U}_{\cat{D}}$, i.e., on tuples of six objects.

We observed that often the objects of the upgrade categories are the same as the objects in the corresponding data categories. But this need not always be the case. One example is the information for updating types in \proteusSrcLang\ which differs from the type environment (which maps type names to types) in the fact that the upgrade data comes as a mapping from type names to pairs of type and type transformer. We are not concerned with types though.
The example for \creol\ also shows that the upgrade component \cat{UD} does not have a correspondent among the data components.

Finally, we decide on the upgrade constructs 
and associate appropriate endo\-functors. 
For \creol\ there are more details to consider than we had for \proteusSrcLang. In \cite{JohnsenOS05dynamicClasses} there is no actual upgrade construct, but only upgrade messages floating in the distributed system and holding the upgrade information. 
Essentially the technique of \cite{JohnsenOS05dynamicClasses} corresponds, in \proteus\ terminology, to a single upgrade construct which appears at every ``ideal'' point in the program and which treats one class at a time. 
The ingenious analysis of the program code of \proteusSrcLang\ can establish at each program point which identifiers can be upgraded without breaking the type safety. 
This preliminary analysis labels each program point with a set of \textit{capabilities}. 
In our situation we can apply the same analysis and use upgrade constructs which are labelled with the set of identifiers that can be safely upgraded at that point:\footnote{One could use the same information to have \textit{incremental upgrades}, where at each point the upgrade is made only for those identifiers which are safe, when possible (dependencies between the names in the upgrade information may not allow for such splitting of the upgrade).} 
\[
S\ \ ::=\ \ \upgradeTermSetKind{\Delta}{c}
\]
where $\Delta$ is a set of class identifiers.
The corresponding upgrade transition rule is:
%
\begin{prooftree}
\AxiomC{\phantom{$\ \transition{U} \ $}}
\UnaryInfC{$\objectTerm{o}{\upgradeTermSetKind{\Delta}{c}} \transition{E_{\Delta}^{c}} \objectTerm{o}{\nilValue}$}
\end{prooftree}
with $E_{\Delta}^{c}\in\morphisms{\catendo{\cat{D}\prodcat\cat{U}_{\cat{D}}}}$ an endo\-functor on the product category from above, which is defined following the work in \cite{JohnsenOS05dynamicClasses}. We need some notation first. 

\begin{definition}[dependencies check]\label{def_dep_check}
We define a binary relation \respectsDependencies\ on partial mappings $\rho,\rho'\in \idClass\rightharpoonup\mathbb{N}$ as:
\[
\rho\respectsDependencies\rho' \mbox{\ \  iff\ \ \ } \forall \identifier{C}\in\idClass: \identifier{C}\in\rho\Rightarrow \identifier{C}\in\rho'\ \wedge\ \rho(\identifier{C})\leq\rho'(\identifier{C}).
\]
\end{definition}

Define the endo\-functor $E_{\Delta}^{c}$ on $\cat{C}\prodcat\cat{A}\prodcat\cat{UN}\prodcat\cat{UC}\prodcat\cat{UA}\prodcat\cat{UD}$ as follows.
\[
E_{\Delta}^{c}(\rho_{c},\rho_{a},\rho_{un},\rho_{uc},\rho_{ua},\rho_{ud})= \phantom{\hspace{45ex}}
\]
\[
\phantom{\hspace{1ex}}
\left\{\begin{array}{l@{\hspace{0ex}}r} 
\left(\begin{array}{l}
\rho_{c}[\identifier{C}\mapsto\rho_{c}(\identifier{C})[\rho_{uc}(\identifier{C})]\mid \forall \identifier{C}\in\Delta\cap\rho_{uc}], \\
\rho_{a}[\identifier{C}\mapsto \rho_{ua}(\identifier{C})\mid\forall \identifier{C}\in\Delta\cap\rho_{ua}], \\
\rho_{un}[\identifier{C}\mapsto \rho_{un}(\identifier{C})+1\mid\forall \identifier{C}\in\Delta\cap(\rho_{uc}\cup\rho_{ua})], \\
\rho_{uc}\setminus\Delta, \\
\rho_{ua}\setminus\Delta, \\
\rho_{ud}\setminus\Delta \\
\end{array} \right) & \begin{array}{l} \mbox{ if }\forall \identifier{C}\in\Delta\cap\rho_{ud}:\\
\ \ \ \ \ \rho_{ud}(\identifier{C})\respectsDependencies\rho_{un}\\ \end{array} \\
%
&\\
(\rho_{c},\rho_{a},\rho_{un},\rho_{uc},\rho_{ua},\rho_{ud}) & \mbox{otherwise} 
\end{array} \right.
\]

The upgrade message used in \cite{JohnsenOS05dynamicClasses} is the special case where $\Delta$ contains one class identifier and the three upgrade objects also contain this single class identifier. The apparent complication in the definition of the endo\-functor comes from the complicated upgrade information that must be manipulated. This has nothing to do with the category theory, but only with the algebraic structures of the underlying objects.
It is easy to check that the above endo\-functor has no sudden jumps. 
We abuse the notation and use set operations between $\Delta$ and mappings $\rho$, referring to the domain of the map. The notation $\rho[\dots]$ denotes the update of the partial map.

Compared to \proteusSrcLang, challenging in the dynamic upgrading mechanism of \creol\ is the fact that the concurrent objects must be upgraded also (i.e., their local attributes), where inheritance would need particular attention, i.e., when a super-class is upgraded in a class hierarchy and objects of a sub-class must be aware of this upgrade. 
%
Objects are the active unit of computation in a distributed
object-oriented setting, and they use messages for
  communication. In \creol\ with upgrades
also the classes are active since they may be changed at runtime. 
The upgrade numbers that the classes keep in the category component
\cat{UN} are used by the objects to upgrade themselves; also objects
keep an upgrade number so to be able to detect when their class type
has been upgraded. In \cite{JohnsenOS05dynamicClasses} upgrading of
the objects, by getting the new attributes, is done in the rewriting
logic implementation through equations.  Equational steps are
  atomic and unobservable, and
  between each rewrite step (transistion) all possible equational steps are performed. This allows class upgrades
  and upgrade numbers to be consistent with the latest upgrade (using
  equations), and older versions of classes are not needed. In
  contrast, the objects themselves may upgrade their state in a
  distributed manner and at different times (using rewrite steps), and
  the objects may run different versions of the class code (according
  to when they last updated). This means that the objects need not be
  aware of the different code versions that other objects run.  The
  \creol\ language and upgrade mechanism are implemented in rewriting
  logic/Maude as an executable prototype, as reported in
  \cite{JohnsenOS05dynamicClasses}, and with support of modular
  program reasoning based on class invariants and communication
  histories. Incremental reasoning is possible when a class upgrade
  respects the old class invariant.

\section{Conclusion and Further Work}\label{sec_conclusion}

We have built on the modular SOS of \cite{Mosses04modularSOS,Mosses99foundationsMSOS} a Dynamic SOS framework which is intended to be used for defining the semantics of dynamic software upgrades.
At the same time we have given modular SOS definitions for concurrent object-oriented programming constructs, where we defined an encapsulating construction on the underlying category theory of MSOS. The encapsulation can be used also in other situations where a notion of localization of the program execution is needed.

We have considered two examples of languages with dynamic software upgrades: the C-like \proteusSrcLang, and the concurrent and distributed object-oriented \creol. 
We have considered the dynamic class upgrades of \creol, as well as the more classical upgrades of \proteusSrcLang.

The upgrade information is externally provided and is not available to the program. This is why the upgrade components cannot be modified nor inspected by the program constructs, unlike the self produced data. The program can only decide upgrade points and what is allowed to be upgraded safely at a point. This is done using the $\upgradeTermSetKind{}{}$ \ constructs which can be automatically inserted in the code using techniques as in \cite{StoyleHBSN07mutatis}. An upgrade allows the program data to be modified in accordance with the available upgrade information. We discard from the upgrade object only the used upgrade information, hence we use an incremental upgrade method. However, this is not fixed and depends on the decision when defining the upgrade endo\-functors.

We have concentrated on the semantic framework, and less on the typing aspects. The cited papers that investigate forms of dynamic upgrade do thorough investigations into typing issues. These investigations can be done over a Dynamic SOS. 
We have exemplified DSOS for the \proteusSrcLang\ language from \cite{StoyleHBSN07mutatis} and discussed the typing aspects. DSOS could be done also for \upgradej\ \cite{BiermanPN08} or \stump\ \cite{NeamtiuH09} since these also adopt the idea of upgrade points.
We have also applied DSOS to the \creol\ language \cite{JohnsenOS05dynamicClasses,JohnsenO07}, where the combination of distributed objects with concurrency and asynchronous method calls with futures, interfaces and inheritance, dynamic binding and behaviour types, make the example non-trivial.



For the question whether DOSO could be encoded solely in the MSOS, mentioned in the introduction,
we see a negative answer because the endo\-functors capture general functions on the objects which cannot readily be captured with the pairs and discrete categories. However, if we use only pairs categories then an encoding seems possible, though how natural it would be is not clear since the morphisms have the computational interpretation of capturing the way data is being manipulated by the program, whereas the endo\-functors encode actions outside the view of the program but which act on the data that the program works with.
At the same time a discrete category can always be replaced by a pairs category without changes to the rules, in which case an encoding seams even more plausible. Thus, this open question seams like a natural immediate continuation of this work.

A programming language designer might also ask whether any dynamic upgrade construct that can be captured by the endo\-functors in the DSOS, can be implemented using the programming language constructs alone.
This question is specific to the programming language and the upgrade mechanism; therefore, it cannot have a general answer at the level of DSOS. For specific situations this seems plausible as long as discrete categories are not used by the program.

\subsection{Possible continuations}\label{subsec_further_work}

A theoretical motivation for giving semantics to dynamic upgrades using DSOS is the close similarity of the transition systems we obtain, with the labelled transition systems obtained by the SOS of process algebras. There is a great wealth of general results in the process algebra community on SOS rule formats \cite{aceto01SOS_handbook}, some of which we hope can be translated to the theory developed here. In particular, the states of the transition systems obtained from DSOS are only program terms, whereas the rest of auxiliary notions are flowing on the transitions as labels. This is the same as in process algebras, only that we have more complex labels. The possible connections between the terms and the structure of the labels in MSOS has been recently investigated in \cite{ChurchillM13fossacs} and endeavours into rule formats with data, like we would need in DSOS, are being investigated \cite{Mousavi07SOSsurvey,Mousavi06SOSdata}. 
General results that could be investigated (starting from the work presented in \cite{aceto01SOS_handbook,ChurchillM13fossacs,Mousavi06SOSdata}) are: 
\begin{enumerate}
\item generating algebraic semantics \cite{AcetoBV94axiomSem,aceto09algProp,Mousavi13algData} from specific forms of the transition rules; 

\item  compositional reasoning results wrt.\ dynamic logic \cite{pratt76semantical,harel00dynamicLogic} using specific forms of transition rules in the style of \cite{FokkinkGW06compositionalGSOS}; or 

\item expressiveness results of the programming constructs specified within various rule formats.
\end{enumerate}
A programming language that is developed within the restrictions of the rule format would get such general results for free.

The modular aspect of Dynamic SOS (and MSOS) is a good motivation for undertaking a more practical challenge of building a database of programming constructs together with their respective (D)MSOS transition rules. A new programming language would then be built by choosing the needed constructs and their preferred semantics, when more exist (e.g., variables implemented with a single store or with a heap and store). The language developer would then only concentrate on the new programming feature/construct that is under investigation.
This was the goal of the PlanComps\footnote{\url{http://plancomps.csle.cs.rhul.ac.uk/}\ \ \ or\ \ \ \url{http://www.plancomps.org}} project which achieved quite significant results \cite{ChurchillMST14PlanCompsTAOSD,BinsbergenSM16MSOS_tools}.
For DSOS we would probably need to extend their results to include the dynamic upgrade semantic concepts of DSOS and also the encapsulation concept. Then all the FunCons of PlanComps would be reusable, and on top would would define similar concepts for dynamic upgrade constructs.

We can mention a few requirements of such a database. 
One is a ready integration of the (D)MSOS rules with a proof assistant like \coqprover, where the work in \cite{pierce12SOS_coq} is a good inspiration point.
Another is the use of a notation format with the possibility of extensible notation style overlays, which would allow the developer to view the semantics in the preferred notation. Nice advancements have been done by people from the PlanComps project, e.g., \cite{ChurchillM13fossacs,ChurchillMM13MSOSformats} as well as relating with the recent K framework \cite{MossesV14MSOSinK,rosu-serbanuta-2010-jlap}.
Such a database needs to be maintainable by the community, as with a wiki.

Another interesting problem is upgrading running code at a more basic level than what \creol\ or \proteusSrcLang\ do where the upgrade happens for methods inside classes and the new execution can be seen only if the currently running code decides to call the upgraded methods; or where types are upgraded and the new code is seen if it is accessed. 
We mean trivial examples like a reactive while loop (i.e., which waits for input from a user to proceed with a round of computation and response) where no methods are called, but where non-trivial computation and checks are done. A bug in such a code (maybe on a branch that is very rarely taken) may be caused by a wrong operation (like plus instead of minus). One wants to correct this running code, and no method or type upgrading would do it. We also do not accept arguments like: put the executing body of the while in a function which is called at each iteration, then upgrade the function when it is finished.

Upgrading such running code could be possible if we view the code as data, having one component of label category keeping track of the current executing code. One could use a program counter variable updated by all execution operations. Kept the program counter together with the actual execution code term $t$,
An upgrade operation of the executing code works with an upgrade component that also contains a new code term $t_{u}$ and an associated new program counter. The execution of the upgrade operation would then replace the execution term with the new one, and the continuing code would be the one given by the new term $t_{u}$ and the associated program counter. The upgrade data for the program may have more complex structure, and the upgrade composition may be more involved than just complete replacing. For example, the new program counter may be depending on the old execution term and the current program counter also; so it may be a function of these. This may well be a map between the possible program counters in the old term $t$ and new program counters in $t_{u}$.

\vspace{2ex}
\noindent\textbf{Acknowledgments:} We are grateful to Martin
Churchill and Martin Steffen (and several anonymous reviewers) for
helping us to improve the paper.
%

\section*{References}

\end{document}